\documentclass[journal,10pt]{IEEEtran} 

\usepackage[T1]{fontenc}

\usepackage{graphicx}
\usepackage[export]{adjustbox}
\usepackage[dvipsnames]{xcolor}
\usepackage{bbm}

\graphicspath{{./Figs/}}
\usepackage{subfigure}
\usepackage{stfloats}
\ifCLASSINFOpdf
\else
\fi

\def\qt{\quad\times}

\def\g{\gamma}

\def\E{\mathbb{E}}
\def\P{\mathbb{P}}

\def\Grml{G_{\mathrm{r,ml}}}
\def\Grsl{G_{\mathrm{r,sl}}}

\def\Nl{N_{v}}

\def\BPP{{\Phi}}
\def\BPPe{\bar{\Phi}}
\def\BPPeml{\bar{\Phi}_{\mathrm{ml}}}
\def\BPPesl{\bar{\Phi}_{\mathrm{sl}}}
\def\PPPe{\bar{\Phi}^{\prime}}
\def\PPPeml{\bar{\Phi}_{\mathrm{ml}}^{\prime}}
\def\PPPesl{\bar{\Phi}_{\mathrm{sl}}^{\prime}}

\def\BPPl{{\Phi_v}}
\def\BPPle{\bar{\Phi}_v}

\def\BPPeEq{\BPPe=\emptyset}

\def\BPPemlEq{\BPPeml=\emptyset}
\def\BPPemlNeq{\BPPeml\neq\emptyset}
\def\BPPeslEq{\BPPesl=\emptyset}
\def\BPPeslNeq{\BPPesl\neq\emptyset}

\def\lame{{\lambda_{\mathrm{e}}}}
\def\thes{{\theta_{\mathrm{s}}}}

\def\as{a_{\mathrm{s}}}
\def\av{a_{\mathrm{e}}}
\def\al{a_{\mathrm{e},v}}

\def\gs{{\gamma_{\mathrm{s}}}}
\def\ge{{\gamma_e}}
\def\gemd{{\gamma_{\mathrm{e}^*}}}
\def\gemdml{{\gamma_{\mathrm{e^*,ml}}^{(p)}}}
\def\gemdsl{{\gamma_{\mathrm{e^*,sl}}^{(q)}}}

\def\glemd{{\gamma_{\mathrm{e}^*,v}}}
\def\gLemd{{\hat{\gamma}_{\mathrm{e}^*}}}

\def\gemdapp{{\tilde{\gamma}_{\mathrm{e}^*}}}
\def\gemdmlapp{{\tilde{\gamma}_{\mathrm{e^*,ml}}}}
\def\gemdslapp{{\tilde{\gamma}_{\mathrm{e^*,sl}}}}

\def\ge0{\g_{\mathrm{e}_0}}
\def\e0{\mathrm{e}_0}

\def\ds{d_{\mathrm{s}}}
\def\dth{d_{\mathrm{th}}}
\def\dmax{d_{\mathrm{max}}}

\def\hs{{h_{\mathrm{s}}}}

\def\cdfgs{F_{\gs}(x)}

\def\pdfgs{f_{\gs}(x)}

\def\Pout{P_{\mathrm{out}}}

\def\PoutL{\hat{P}_{\mathrm{out}}}

\def\Rt{R_{\mathrm{t}}}
\def\Cerg{C_{\mathrm{erg}}}
\def\Cout{C_{\mathrm{out}}}
\def\CergL{\hat{C}_{\mathrm{erg}}}

\def\X{\mathcal{X}}

\def\A{\mathcal{A}}
\def\Ae{\bar{\mathcal{A}}}
\def\Aeml{\bar{\mathcal{A}}_{\mathrm{ml}}}
\def\Aesl{\bar{\mathcal{A}}_{\mathrm{sl}}}
\def\idxe{{\bar{e}}}
\def\idxet{{\bar{e}}}

\def\S{\mathcal{S}}
\def\SA{\mathcal{S}_{\A}}
\def\SAe{\mathcal{S}_{\Ae}}
\def\SAeml{\mathcal{S}_{\Aeml}}
\def\SAesl{\mathcal{S}_{\Aesl}}

\def\psmax{\psi_{\mathrm{max}}}
\def\psth{\psi_{\mathrm{th}}}

\def\wth{\omega_{\mathrm{th}}}
\def\wsb{\Delta\omega_{\mathrm{sb}}}

\usepackage{amssymb}
\usepackage{amsmath}
\usepackage[cmintegrals]{newtxmath}
\usepackage{stackengine}
\def\delequal{\mathrel{\ensurestackMath{\stackon[1pt]{=}{\scriptstyle\Delta}}}}

\usepackage{algorithmic}
\usepackage{algorithm}

\usepackage{amsthm}
\usepackage{xpatch}

\newtheorem{thm}{Theorem}
\newtheorem{lem}{Lemma}

\newtheorem{cor}{Corollary}
\newtheorem{rem}{Remark}

\xpatchcmd{\proof}{\hskip\labelsep}{\hskip5\labelsep}{}{}

\usepackage{mathtools}

\DeclarePairedDelimiter{\floor}{\lfloor}{\rfloor}
\newcounter{myeqncount}

\usepackage{array}

\begin{document}
\title{When Satellites Work as Eavesdroppers}

\author{Dong-Hyun Jung, Joon-Gyu Ryu, and Junil Choi\\

\thanks{This work was supported by Institute of Information \& communications Technology Planning \& Evaluation (IITP) grant funded by the Korea government (MSIT) (No.2021-0-00847, Development of 3D Spatial Satellite Communications Technology). \textit{(Corresponding author: Junil Choi)}}
\thanks{D.-H. Jung is with the School of Electrical Engineering, KAIST, and with the Radio and Satellite Research Division, Communication and Media Research Laboratory, Electronics and Telecommunications Research Institute, Daejeon, South Korea (e-mail: donghyunjung@kaist.ac.kr).}
\thanks{J.-G. Ryu is with the Radio and Satellite Research Division, Communication and Media Research Laboratory, Electronics and Telecommunications Research Institute, Daejeon, South Korea (e-mail: jgryurt@etri.re.kr).}
\thanks{J. Choi is with the School of Electrical Engineering, KAIST, Daejeon, South Korea (e-mail: junil@kaist.ac.kr).}
\vspace{-0.7cm}
}
\maketitle

\begin{abstract}
This paper considers \textit{satellite eavesdroppers} in uplink satellite communication systems where the eavesdroppers are randomly distributed at arbitrary altitudes according to homogeneous binomial point processes and attempt to overhear signals that a ground terminal transmits to a serving satellite.
Non-colluding eavesdropping satellites are assumed, i.e., they do not cooperate with each other, so that their received signals are not combined but are decoded individually.
Directional beamforming with two types of antennas: fixed- and steerable-beam antennas, is adopted at the eavesdropping satellites.
The possible distribution cases for the eavesdropping satellites and the distributions of the distances between the terminal and the satellites are analyzed.
The distributions of the signal-to-noise ratios (SNRs) at both the serving satellite and the most detrimental eavesdropping satellite are derived as closed-form expressions.
The ergodic and outage secrecy capacities of the systems are derived with the secrecy outage probability using the SNR distributions.
Simpler approximate expressions for the secrecy performance are obtained based on the Poisson limit theorem, and asymptotic analyses are also carried out in the high-SNR regime.
Monte-Carlo simulations verify the analytical results for the secrecy performance.
The analytical results are expected to be used to evaluate the secrecy performance and design secure satellite constellations by considering the impact of potential threats from malicious satellite eavesdroppers.

\\
\end{abstract}

\IEEEpeerreviewmaketitle

\vspace{-0.6cm}

\section{Introduction}\label{Sec:Intro}
\IEEEPARstart{S}{atellite}
communications have been recently employed to provide global internet services exploiting the large coverage of satellites.
As a part of the fifth generation standard, non-terrestrial networks (NTNs) have been standardized by 3rd Generation Partnership Project (3GPP) since Release 15 [\ref{Ref:3GPP_38.811}], [\ref{Ref:3GPP_38.821}].
The 3GPP standard is to harmonize the original terrestrial networks (TNs) with the NTNs composed of flying entities such as satellites, unmanned aerial vehicles (UAVs), and high-altitude platforms (HAPs).
Such three-dimensional networks including both TNs and NTNs will provide communication services to not only ground terminals but those in the sky such as drones, airplanes, and vehicles for urban air mobility.
In addition, several companies such as SpaceX, OneWeb, Telesat, and Amazon have planned to launch a large number of low Earth orbit (LEO) satellite constellations to enhance the system throughput in the near future [\ref{Ref:Pachler}]. 
For example, SpaceX has plans to launch more than 10,000 satellites over the next couple of years, while OneWeb designed a constellation with 720 satellites at the altitude of 1,200 km.
Telesat has planned to construct the constellation of 117 satellites on the polar and inclined orbits at the altitude of 1,000 and 1,200 km, respectively.
However, the growing number of satellites may cause problems from a communication security perspective. 
For example, some satellites with specific purposes could overhear important information transmitted from ground nodes without permission.
With these potential threats by the massive number of satellites, it is important to keep communication systems safe in terms of communication security.

\vspace{-0.3cm}
\subsection{Related Works}\label{Sec:Intro_rel_work}
% \subsection{PLS concept + Sat/NTN}
Due to the characteristics of wireless channels, the information signals from a transmitter can reach not only a legitimate receiver but other malicious receivers, which act as eavesdroppers [\ref{Ref:Wyner}].
With the knowledge of the communication protocols of the legitimate link, the eavesdroppers are able to overhear the signals without permission.
The positions of eavesdroppers may not be known at legitimate transceivers, as the eavesdroppers usually act as passive nodes.
To tackle this uncertainty, many works considered multiple eavesdroppers randomly distributed in infinite areas based on Poisson point processes (PPPs) [\ref{Ref:Geraci}]-[\ref{Ref:Chen2}].
However, when the number of nodes distributed in the networks is finite, the randomness of the positions of the nodes should be modeled by using a finite point process other than the PPP [\ref{Ref:Afshang}].

The binomial point process (BPP) is a finite point process where each point is independently distributed, and the number of points in a bounded region follows the binomial distribution.
The distributions of LEO satellite constellations can be modeled as the BPP over a sphere-shaped region because the satellites may look randomly distributed due to the mobility of the LEO satellites and various orbits of the practical constellations [\ref{Ref:Okati}]-[\ref{Ref:Talgat2}]. It was shown in [\ref{Ref:Okati}] that the BPP is appropriate to model practical LEO satellites' distributions from coverage and rate perspectives.
The user coverage probability of LEO satellite communication systems was studied in [\ref{Ref:Talgat1}] where gateways act as relays between users and LEO satellites.
The distance distributions for gateway-satellite and inter-satellite links were studied in [\ref{Ref:Talgat2}].

% NTN secrecy
As interest in the NTNs including satellites, UAVs, and HAPs is growing, physical layer security in the NTNs has recently been investigated in [\ref{Ref:Zhu}]-[\ref{Ref:Wu}] assuming ground eavesdroppers.
The ergodic secrecy capacity for UAV networks was analyzed in [\ref{Ref:Zhu}] where a transmit jamming strategy was proposed to confuse randomly-located eavesdroppers. 
Zero-forcing-based beamforming schemes for multi-beam satellite systems were proposed in [\ref{Ref:Lei}] and [\ref{Ref:Zheng1}] to minimize the transmit power of the satellite with a secrecy rate constraint.
For both perfect and imperfect channel state information, secure transmission schemes in cognitive satellite-terrestrial networks were proposed in [\ref{Ref:Li}].
As stated, however, these works are based only on the ground eavesdroppers in the NTNs.

\subsection{Motivation and Contributions}\label{sec:Intro_motiv_cont}
As the number of NTN elements, such as satellites and UAVs, rapidly grows in the sky, some of them may illegally attempt to eavesdrop on information signals from ground transmitters for political or military purposes. In addition, the large coverage, which is an advantage of the NTNs, may rather become a major threat, making it easier for the signals to be overheard.
The impacts of UAV eavesdroppers on the secrecy performance have been recently investigated in [\ref{Ref:Sharma}]-[\ref{Ref:Bao}].
The secrecy outage probability of hybrid satellite-terrestrial networks was analyzed in [\ref{Ref:Sharma}] where a UAV eavesdropper tries to overhear information signals from a UAV relay.
Asymptotic expressions for the ergodic and outage secrecy capacities were derived in [\ref{Ref:Yuan1}] with a UAV eavesdropper.
A secure connection probability was analyzed in [\ref{Ref:Tang}] in the presence of UAV eavesdroppers randomly located in a disk according to a BPP. 
For both colluding and non-colluding aerial eavesdroppers that are distributed over a hemisphere, the ergodic and outage secrecy capacities were analyzed in [\ref{Ref:Yuan2}].
The secrecy outage probability of ground-to-air communication networks was derived in [\ref{Ref:Bao}] where a deep learning model to predict the secrecy performance was developed.
However, there have been few studies on the impact of satellites working as eavesdroppers, which are worthwhile to investigate as the number of satellites increases.

Motivated by this, we aim to analyze the secrecy performance of satellite communication systems in the presence of multiple satellite eavesdroppers.
% , distributed according to a BPP.
The main contributions of this paper are summarized as follows.
\begin{itemize}
    \item \textbf{Eavesdropping satellites:} The impact of satellites working as eavesdroppers is investigated, while previous works on the physical layer security in satellite communication systems [\ref{Ref:Lei}]-[\ref{Ref:Kalantari}] only consider the eavesdroppers located on the Earth.
    Different from aerial UAV eavesdroppers in [\ref{Ref:Tang}]-[\ref{Ref:Bao}], satellites are distributed over a sphere centered at the Earth's center, which makes distance distributions different. 
    \item \textbf{Two types of beamforming antennas:} 
    Directional beamforming with both fixed- and steerable-beam antennas is considered for the eavesdropping satellites to compensate the large path loss [\ref{Ref:3GPP_38.821}], which has not been considered in the previous works [\ref{Ref:Tang}]-[\ref{Ref:Bao}].
    The satellites with the fixed-beam antennas maintain the boresight fixed in the direction of the subsatellite point (the nearest point on the Earth), while with the steerable-beam antennas, the boresight direction of the satellite's beam can be steered to the location of the targeted terminal to improve the quality of received signals.
    \item \textbf{BPP-based secrecy performance analyses:} The possible distribution cases for eavesdropping satellites are studied, and the probabilities of these cases are derived based on the mathematical properties of the BPP.
    The distributions of distances between the terminal and eavesdropping satellites are analyzed.
    With the derived distance distributions, we obtain the distributions of the signal-to-noise ratios (SNRs) at the serving satellite and the most detrimental eavesdropping satellite, i.e., the eavesdropping satellite with the highest received SNR. 
    For both fixed- and steerable-beam antennas, we analyze three secrecy performance: (i) ergodic secrecy capacity, (ii) secrecy outage probability, and (iii) outage secrecy capacity of the systems by using the SNR distributions. 
    Using the Poisson limit theorem, we derive simpler approximate expressions for the secrecy performance and compare their computational complexity with the exact expressions. In addition, the asymptotic behavior of the secrecy performance is investigated in the high-SNR regime.
    \item \textbf{Verification by simulations:} Finally, we numerically verify the derived expressions through Monte-Carlo simulations.
    The impact of the number of eavesdropping satellites, position of the serving satellite, and beam-steering capability is also shown in the simulations. 
\end{itemize}

The rest of this paper is organized as follows.
In Section \ref{Sec:System_model}, the system model for satellite communication systems with randomly-located eavesdropping satellites is described.
In Section \ref{Sec:BPP-based_Anal}, the possible distribution cases and the distance distributions for the eavesdropping satellites are studied.
In Section \ref{Sec:Secrecy_performance_anal},the SNR distributions and the secrecy performance of the systems are analyzed.
In Section \ref{Sec:Asym_secrecy_performance}, the approximate and asymptotic secrecy performance are obtained.
In Section \ref{sec:sat_diff_alti}, the analyzed results are extended to the satellites at different altitudes.
In Section \ref{Sec:Sim_results}, simulation results are provided, and conclusions are drawn in Section~\ref{Sec:Conclusions}.

\emph{Notation:}
$\mathbb{P}[\cdot]$ indicates the probability measure, and $\mathbb{E}[\cdot]$ denotes the expectation operator.
The empty set is $\emptyset$, and the complement of a set $\mathcal{X}$ is $\mathcal{X}^{\mathrm{c}}$.
The length between two points $\mathrm{A}$ and $\mathrm{B}$ is $\overline{\mathrm{AB}}$.
The absolute value of a real number $x$ is $|x|$.
% The surface area of a region $\mathcal{X}$ is $\S_\mathcal{X}$.
The indicator function is $\mathbf{1}(x\in\mathcal{X})$, which has the value of 1 if $x\in \mathcal{X}$ and 0 otherwise.
$\binom{n}{k}$ denotes the binomial coefficient.
The cumulative distribution function (CDF) and the probability density function (PDF) of random variable $X$ are $F_X(x)$ and $f_x(x)$, respectively.
$\Gamma(\cdot)$ is the Gamma function, and the Pochhammer symbol is defined as $(x)_n=\Gamma(x+n)/\Gamma(x)$.
The lower incomplete Gamma function is defined as $\gamma(a, x)=\int_0^x t^{a-1}e^{-t}dt$.
The inverse function of $f(\cdot)$ is $f^{-1}(\cdot)$.
The floor function is $\floor{\cdot}$.

\begin{figure*}[!t]
\centering
\subfigure[]{
\includegraphics[width=0.855\columnwidth]{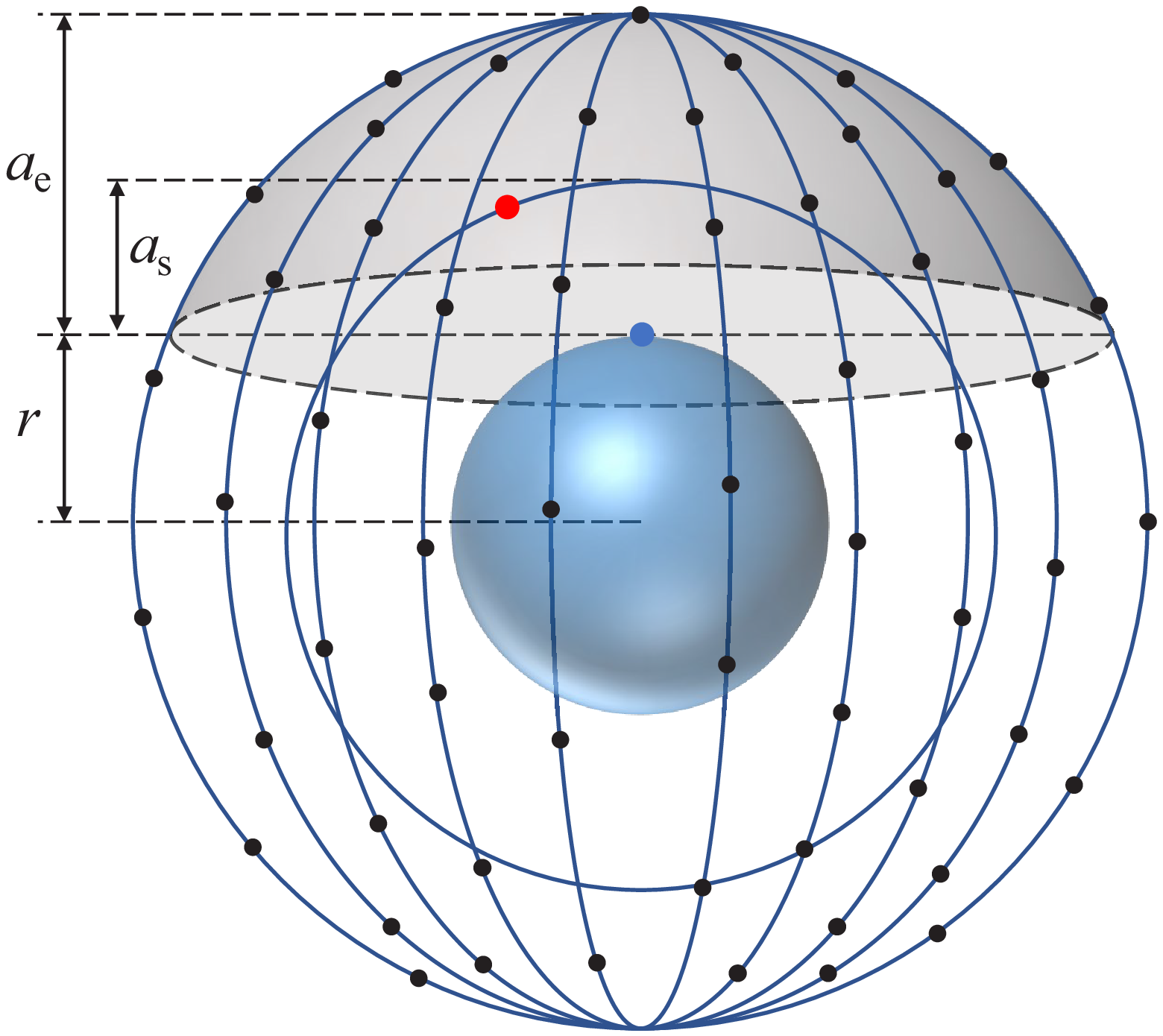}
\label{Fig:system_model_1}
}
\subfigure[]{
\includegraphics[width=0.855\columnwidth]{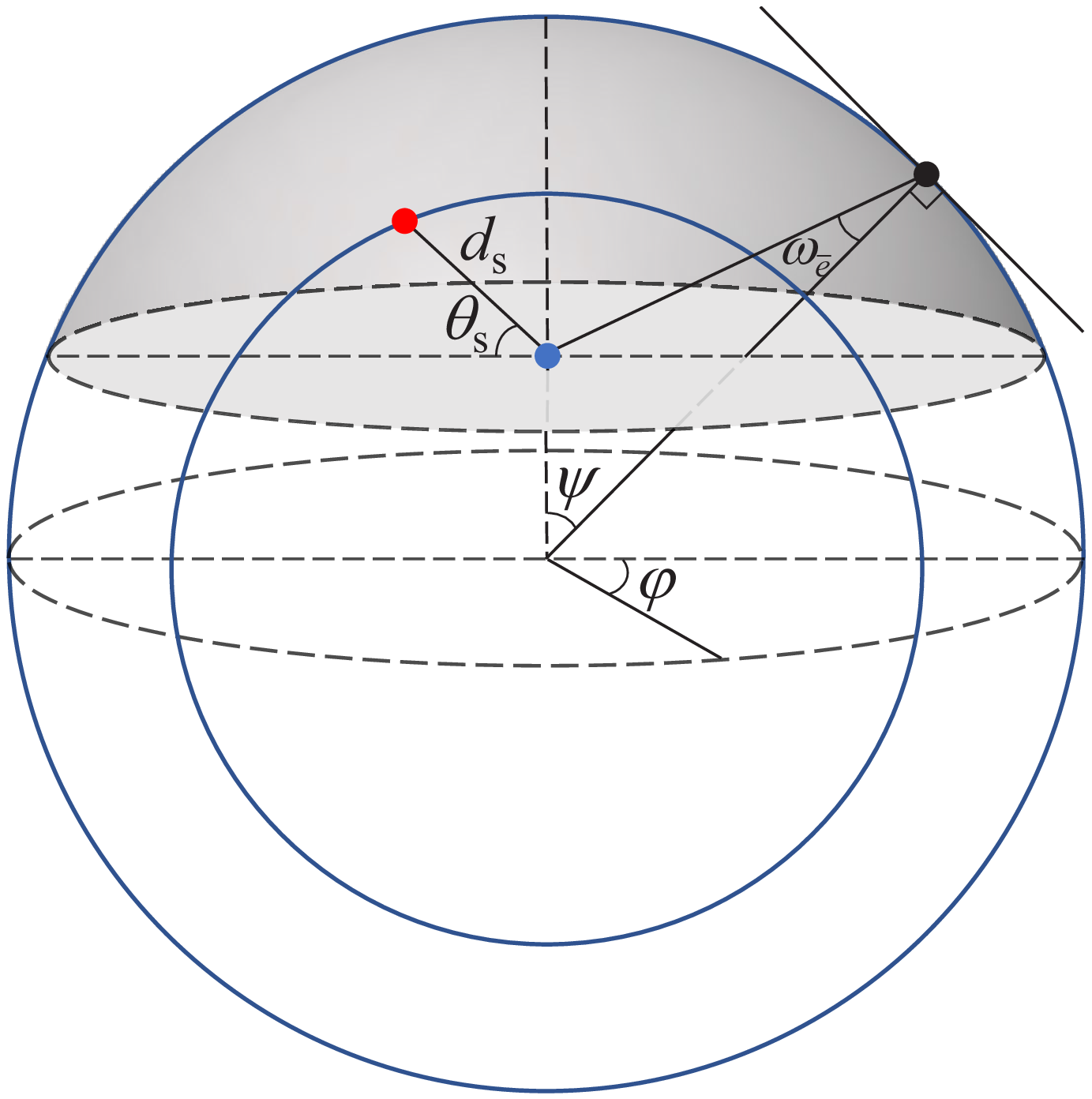}
\label{Fig:system_model_2}
}
\caption{(a) System model and (b) parameter description. Blue, red, and black dots indicate the terminal, serving satellite, and eavesdropping satellites, respectively.}
\label{Fig:System_model}
\vspace{-0.3cm}
\end{figure*}

\vspace{-0.2cm}
\section{System Model}\label{Sec:System_model}
Consider an uplink LEO satellite communication system where a ground terminal communicates to a serving satellite at the altitude $\as$ and the elevation angle $\thes$ in the presence of multiple eavesdropping satellites. 
Assume the handheld-type terminal that is equipped with an omni-directional antenna as considered in the 3GPP NTN standard.\footnote{The 3GPP NTN standard considers the handheld terminal, i.e., the user equipment with the power class 3, with an omni-directional antenna as one of the target terminals [\ref{Ref:3GPP_38.821}]. This antenna configuration can be applicable to low-cost terminals, such as handheld and IoT terminals.} 
We begin with the assumption that $N$ eavesdropping satellites are randomly distributed at the altitude $\av$ according to a homogeneous BPP ${\Phi}$, as shown in Fig. \ref{Fig:system_model_1}, which will be extended to the general case of eavesdropping satellites at different altitudes in Section \ref{sec:sat_diff_alti}.
The region where the eavesdropping satellites can be located is a sphere with the radius $r+\av$, where $r$ is the radius of the Earth.
With spherical coordinates, this region can be expressed as $\A=\{\rho=r+\av, 0\le\psi\le\pi, 0 \le\varphi\le 2\pi\}$, where $\rho$, $\psi$, and $\varphi$ are the radial distance, polar angle, and azimuthal angle, respectively. 
Among the eavesdropping satellites in $\A$, those located above the horizontal plane at the terminal are only considered as effective eavesdropping satellites, which attempt to overhear the information signals transmitted from the terminal.
Let $\BPPe$ denote the set of the effective eavesdropping satellites, and $\Ae$ represent the region where the satellites in $\BPPe$ can be located,  shown as the shaded areas in Fig. \ref{Fig:system_model_2}.
Since $\Ae$ is the surface of the spherical cap with the radius $r+\av$ above the horizontal plane at the terminal, $\Ae$ can be expressed as $\Ae = \{\rho=r+\av, 0\le\psi\le\psmax, 0 \le\varphi\le 2\pi\}$, where the maximum polar angle is given by $\psi_{\mathrm{max}}=\cos^{-1}\left(\frac{r}{r+\av}\right)$.
Assume that the eavesdropping satellites are non-colluding, i.e., they do not cooperate with each other.\footnote{Colluding satellite eavesdroppers may not be practical due to the limited capacities of inter-satellite and feeder links. In addition, the signals transmitted from the eavesdroppers experience very different propagation delays and Doppler offsets due to geographical distance between the satellites. From these facts, a collecting node may have difficulty combining the received signals.}
 
Directional beamforming is adopted at both the serving and eavesdropping satellites to compensate the path loss.
Assume that the main lobe of the serving satellite's beam is directed towards the terminal based on the terminal's location.
For the eavesdropping satellites, two types of beamforming capabilities are considered: (i) satellites with fixed-beam antennas, which maintain the boresight in the direction of the subsatellite point and (ii) satellites with steerable-beam antennas, of which the boresight direction can be steered to a targeted position.

For mathematical tractability, the satellites are assumed to have a sectorized beam pattern whose antenna gains for the main and side lobes are respectively given by $G_{\mathrm{r,ml}}$ and $G_{\mathrm{r,sl}}$ [\ref{Ref:Alkhateeb}].
Let $\omega_i$, $i \in \{\mathrm{s}, \idxe\}$, $\idxe\in \BPPe$, denote the angle between the terminal and the boresight direction of the satellite $i$, where the indices $\mathrm{s}$ and $\idxe$ denote the serving and effective eavesdropping satellites, respectively.
Then, the receive antenna gain of the satellite $i$ is given by
\begin{align}
G_{\mathrm{r},i}
    =& \begin{cases} 
    G_{\mathrm{r,ml}}, & \mbox{if  } |\omega_{i}| \le \wth,\\
    G_{\mathrm{r,sl}}, & \mbox{otherwise}, 
    \end{cases}
\end{align}
where $\wth$ is the threshold angle between the main and side lobes of the beam pattern.
The propagation loss for the link from the terminal to the satellite $i$ is modeled as $\ell(d_i)= G_{\mathrm{t}} G_{\mathrm{r},i} \left(\frac{c}{4\pi f_{\mathrm{c}}}\right)^2 d_i^{-\alpha}$, where $d_i$ is the distance between the terminal and the satellite $i$, $G_{\mathrm{t}}$ is the transmit antenna gain, $c$ is the speed of light, $f_{\mathrm{c}}$ is the carrier frequency, and $\alpha$ is the path-loss exponent.
The shadowed-Rician fading model is assumed for the channels between the terminal and satellites.\footnote{The shadowed-Rician fading model is widely adopted for satellite channels, which has been proved its suitability in various frequency bands, e.g., the UHF-band, L-band, S-band, and Ka-band [\ref{Ref:Jung}]-[\ref{Ref:Bhatnagar2}].}
Let $h_{i}$ denote the channel gain between the terminal and the satellite $i$. Then, the CDF of the channel gain is given by [\ref{Ref:Abdi}]
\begin{equation}\label{eq:CDF_ch_gain}
{F_{h_{i}}}(x) = K\sum\limits_{n = 0}^\infty  {\frac{{{{(m)}_n}{\delta^n}{{(2b)}^{1 + n}}}}{{{{(n!)}^2}}}}\gamma\left(1+n,\frac{x}{2b}\right),
\end{equation}
where $K = {\left({2bm}/{(2bm + \Omega) }\right)^m}/{2b}$, $\delta  = (\Omega /(2bm + \Omega ))/2b$ with $\Omega$ being the average power of the line-of-sight (LOS) component, $2b$ is the average power of the multi-path components except the LOS component, and $m$ is the Nakagami parameter.

The SNR at the satellite $i$, $i \in \{\mathrm{s}, \idxe\}$, $\idxe \in \BPPe$, is given by
\begin{equation}\label{eq:SNR}
\g_i =  \frac{P h_i \ell(d_i)}{N_0 W},
\end{equation}
where $P$ is the transmit power of the terminal, $N_0$ is the noise power spectral density, and $W$ is the bandwidth.
For the non-colluding satellites, the secrecy rate of the system is determined by the most detrimental eavesdropping satellite, i.e., the eavesdropping satellite with the highest SNR, which is given by
\begin{align}\label{eq:SNRe*}
\gemd
    =& \begin{cases} 
    \max\limits_{\idxe\in\BPPe} \g_\idxe, & \mbox{if } \BPPe \neq \emptyset,\\
    0, & \mbox{if } \BPPe = \emptyset.
     \end{cases}
\end{align}
Since the eavesdropping satellites are randomly distributed over a sphere, it is possible that no eavesdropping satellites are located in $\Ae$, i.e., $\BPPe=\emptyset$, which means that there is no malicious node to eavesdrop on the signals from the terminal.
In such case, the SNR at the most detrimental eavesdropping satellite is considered to be zero.

\section{Binomial Point Process-Based Analyses}\label{Sec:BPP-based_Anal}
In this section, the possible distribution cases and distance distributions for the eavesdropping satellites are analyzed based on the characteristics of the BPP.
Hereafter, we analyze secrecy performance assuming the eavesdropping satellites with the fixed-beam antennas and then extend the results to those with the steerable-beam antennas in Section~\ref{Sec:Secrecy_performance_anal_sb}.

\subsection{Distribution Cases for Eavesdropping Satellites}
When the eavesdropping satellites with the fixed-beam antennas maintain the boresight in the direction of the subsatellite point, the set of effective eavesdropping satellites $\BPPe$ can be divided into two sets $\BPPeml$ and $\BPPesl$ consisting of the effective eavesdropping satellites whose main and side lobes are directed towards the terminal, respectively.
The regions where the satellites in $\BPPeml$ and $\BPPesl$ can be located are denoted by $\Aeml$ and $\Aesl$, respectively, which are shown in Fig. \ref{Fig:area_description}. 
The threshold polar angle $\psth$ differentiating $\Aeml$ and $\Aesl$ can be obtained as
\begin{align} \label{eq:OP}
\overline{\mathrm{OP}} 
    =&\overline{\mathrm{OS}}\sin\wth=\overline{\mathrm{OT}}\sin(\wth+\psth)\nonumber\\   
    =& (r+\av)\sin\wth=r\sin(\wth+\psth),
\end{align}
which gives the threshold polar angle as
\begin{align}
\psth
    = \sin^{-1}\left(\frac{r+\av}{r} \sin\wth\right)-\wth,
\end{align}
when $\wth \le \sin^{-1}\left(\frac{r}{r+\av}\right)$.
If $\wth > \sin^{-1}\left(\frac{r}{r+\av}\right)$, the main lobes of all the effective eavesdropping satellites are directed toward the terminal so that $\Aeml$ is expanded to $\Ae$, i.e., $\psth=\psmax$.
With respect to $\psth$, the regions $\Aeml$ and $\Aesl$ can be expressed as $\Aeml= \{\rho=r+\av, 0\le\psi\le\psth, 0 \le\varphi\le 2\pi\}$ and $\Aesl= \{\rho=r+\av, \psth<\psi\le\psi_{\mathrm{max}}, 0 \le\varphi\le 2\pi\}$,
respectively.
As the satellites are randomly distributed over a sphere, the numbers of satellites in $\Aeml$ and $\Aesl$ can be zero, i.e., $\BPPemlEq$ and/or $\BPPeslEq$. 
Hence, there can be four distribution cases for the eavesdropping satellites as follows.
\begin{itemize}
    \item Case 1: There are no effective eavesdropping satellites, i.e., $\BPPeEq$, or equivalently $\BPPemlEq$ and $\BPPeslEq$.
    \item Case 2: There are no eavesdropping satellites in $\Aeml$, but in $\Aesl$, i.e., $\BPPemlEq$ and $\BPPeslNeq$.
    \item Case 3: There are no eavesdropping satellites in $\Aesl$, but in $\Aeml$, i.e., $\BPPemlNeq$ and $\BPPeslEq$.
    \item Case 4: At least one eavesdropping satellite exists in both $\Aeml$ and $\Aesl$, i.e., $\BPPemlNeq$ and $\BPPeslNeq$.
\end{itemize}
Before driving the probabilities of the four cases, we obtain the finite-dimensional distribution of the BPP $\BPP$ in the following lemma.

\begin{figure}
\begin{center}
\includegraphics[width=0.95\columnwidth]{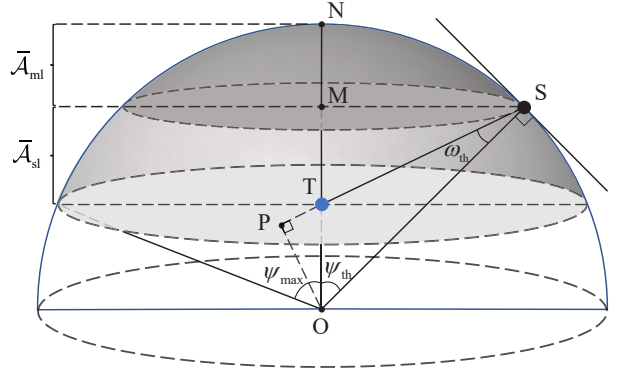}
\end{center}
\setlength\abovecaptionskip{.25ex plus .125ex minus .125ex}
\setlength\belowcaptionskip{.25ex plus .125ex minus .125ex}
\caption{Descriptions of the areas $\Aeml$ and $\Aesl$ where the main and side lobes of the satellites' beams are respectively directed toward the terminal. The points O, T, and S respectively represent the positions of the Earth's center, the terminal, and the eavesdropping satellite at the boundary between $\Aeml$ and $\Aesl$.}
\label{Fig:area_description}
\end{figure}

\begin{lem}\label{Lem:finite_dim_dist}
Let $\BPP(\mathcal{X})$ denote a random variable representing the number of eavesdropping satellites of the BPP $\BPP$ lying in a region $\mathcal{X}$, i.e., $\BPP(\mathcal{X})=\sum_{\idxe\in\BPP}\mathbf{1}(\idxe\in\mathcal{X})$. 
For disjoint regions $\X_1$, $\X_2$, $\cdots$, $\X_J$ satisfying $\X_1\cup\X_2\cup\cdots\cup\X_J=\A$, 
the finite-dimensional distribution, i.e., the probability distribution of the random vector $[\BPP(\X_1) \,\, \BPP(\X_2)\,\, \cdots\,\, \BPP(\X_J)]$, is given by
\begin{align}
% \mathcal{P}(&n_1,n_2,\cdots, n_{K-1})\nonumber\\
%     &=
\mathbb{P}[\BPP(&\X_1)=n_1,\BPP(\X_2)=n_2, \cdots, \BPP(\X_J)=n_J] \nonumber\\
    &= \frac{N!}{n_1!n_2!\cdots n_J!} \times \frac{(\S_{\X_1})^{n_1} (\S_{\X_2})^{n_2} \cdots (\S_{\X_J})^{n_J}}{(\S_\A)^N},
\end{align}
where $\mathcal{S}_{\mathcal{X}_j}$ is the surface area of the region $\mathcal{X}_j$.
\end{lem}

\begin{proof}[Proof:\nopunct]
This result is given in [\ref{Ref:Book:Chiu}].
\end{proof}

From Lemma \ref{Lem:finite_dim_dist}, we obtain the probability distribution for the three disjoint regions $\Aeml$, $\Aesl$, and $\Ae^{\mathrm{c}}$ in the following corollary.

\begin{cor}\label{Cor:prob_pq}
The probability distribution for the three disjoint regions $\Aeml$, $\Aesl$, and $\Ae^{\mathrm{c}}$ is given by
\begin{align}\label{eq:prob_pq_fin}
\mathcal{P}[N,p,q]&\triangleq
    \mathbb{P}[\BPP(\Aeml)=p,\BPP(\Aesl)=q, \BPP(\Ae^{\mathrm{c}})=N\!-\!p\!-\!q] \nonumber\\
    &= \frac{N!}{p!q!(N-p-q)!}\left(\frac{1-\cos\psth}{2}\right)^p\nonumber\\
    &\quad\times\left(\frac{(r+\av)\cos\psth-r}{2(r+\av)}\right)^q\! \left(1\!-\!\frac{\av}{2(r+\av)}\right)^{N-p-q}.
\end{align}
\end{cor}

\begin{proof}[Proof:\nopunct]
See Appendix \ref{App:Cor1}.
\end{proof}

In \eqref{eq:prob_pq_fin}, we define a new function $\mathcal{P}[N,p,q]$ for simplicity of notation since the probability distribution becomes a constant with given $N$, $p$, and $q$.
Using Corollary 1, the probabilities of the four cases are obtained as
$P_1=\mathcal{P}[N,0,0]$, $P_2=\sum_{p=1}^{N}\mathcal{P}[N,p,0]$, $P_3=\sum_{q=1}^{N}\mathcal{P}[N,0,q]$, and $P_4=\sum_{p=1}^{N}\sum_{q=1}^{N-p}\mathcal{P}[N,p,q]$,
respectively, where $\sum_{i=1}^{4}P_i = \sum_{p=0}^{N}\sum_{q=0}^{N-p}\mathcal{P}[N,p,q] = 1$.

\subsection{Distance Distributions}
In this subsection, we characterize the distributions of the distances from the terminal to the satellites, which is an important step to derive secrecy performance of the system.
\begin{lem}\label{Lem:CDF_d_e}
The CDF of the distance between the terminal and eavesdropping satellite $\idxe\in\BPP$ is given by 
\begin{align}\label{eq:CDFde}
F_{d_\idxe}(x)=& \begin{cases} 0, & \mbox{if } x \le \av,\\ 
     \frac{x^2-\av^2}{4r (r+\av)}, & \mbox{if } \av < x \le 2r+\av, \\
     1, & \mbox{if } x > 2r+\av.
     \end{cases}
\end{align}
\end{lem}

\begin{proof}[Proof:\nopunct]
See Appendix \ref{App:Lem2}.
\end{proof}

Let $X_\idxe$ denote the distance between the terminal and eavesdropping satellite $\idxe$, given $\idxe \in \BPPeml$.
Then, the CDF and PDF of $X_\idxe$ are obtained in the following two lemmas using the result in Lemma \ref{Lem:CDF_d_e}.

\begin{lem}\label{Lem:CDF_X_e}
The CDF of $X_\idxe$ is given by 
\begin{align}\label{eq:CDF_Xe}
F_{X_\idxe}(x)=& \begin{cases} 0, & \mbox{if } x \le \av,\\ 
     \frac{x^2-\av^2}{\dth^2-\av^2}, & \mbox{if } \av < x \le \dth, \\
     1, & \mbox{if } x > \dth,
     \end{cases}
\end{align}
where $\dth=\sqrt{r^2+(r+\av)^2-2r(r+\av)\cos\psth}$.
\end{lem}

\begin{proof}[Proof:\nopunct]
See Appendix \ref{App:Lem3}.
\end{proof}

\begin{lem}\label{Lem:PDF_Xe}
The PDF of $X_\idxe$ is given by 
\begin{align}\label{eq:PDF_Xe}
f_{X_\idxe}(x)=& \begin{cases} \frac{2x}{\dth^2-\av^2}, & \mbox{if } \av < x \le \dth, \\
     0, & \mbox{otherwise}.
     \end{cases}
\end{align}
\end{lem}
\begin{proof}[Proof:\nopunct]
The proof is complete by differentiating \eqref{eq:CDF_Xe}.
\end{proof}

Given $\idxe \in \BPPesl$, the CDF and PDF of the distance between the terminal and eavesdropping satellite $\idxe$, denoted by $Y_\idxe$, are provided next.

\begin{lem}\label{Lem:CDF_Ye}
The CDF of $Y_\idxe$ is given by 
\begin{align}\label{eq:CDF_Ye}
F_{Y_\idxe}(x)
    =& \begin{cases} 0, & \mbox{if } x \le \dth,\\ 
     \frac{x^2-\dth^2}{\dmax^2-\dth^2}, & \mbox{if } \dth < x \le \dmax, \\
     1, & \mbox{if } x > \dmax,
     \end{cases}
\end{align}
where $\dmax=\sqrt{\av(2r+\av)}$.
\end{lem}

\begin{proof}[Proof:\nopunct]
The proof is similar to that of Lemma \ref{Lem:CDF_X_e}.
\end{proof}

\begin{lem}\label{Lem:PDF_Y_e}
The PDF of $Y_\idxe$ is given by 
\begin{align}\label{eq:PDF_Ye}
f_{Y_\idxe}(x)=& \begin{cases} \frac{2x}{\dmax^2-\dth^2}, & \mbox{if } \dth < x \le \dmax, \\
     0, & \mbox{otherwise}.
     \end{cases}
\end{align}
\end{lem}

\begin{proof}[Proof:\nopunct]
The PDF of $Y_\idxe$ is obtained by differentiating \eqref{eq:CDF_Ye}.
\end{proof}

The finite-dimensional distribution and the distance distributions for the eavesdropping satellites will be used to derive the secrecy performance in the following section.

\section{Secrecy Performance Analyses}\label{Sec:Secrecy_performance_anal}
In this section, the secrecy performance of the system is analyzed as follows.
We first derive the CDFs of the SNRs at the serving satellite and the most detrimental eavesdropping satellite. 
Using the results in the previous section with the CDFs of the SNRs, we derive the analytical expressions for the ergodic and outage secrecy capacities with the secrecy outage probability.
Finally, we extend the results to the eavesdropping satellites with the steerable-beam antennas.

\subsection{SNR Distributions}
To analyze the secrecy performance, we need the distributions of the SNRs at the serving and the most detrimental eavesdropping satellites, which are obtained in the following lemmas. 
\begin{lem}\label{Lem:CDF_SNRs}
The CDF of the SNR at the serving satellite located at the altitude $\as$ with the elevation angle $\thes$ is given by 
\begin{align}\label{eq:CDF_SNRs_fin}
\cdfgs 
    = K\sum\limits_{n = 0}^\infty  {\frac{{{{(m)}_n}{\delta ^n}{{(2b)}^{1 + n}}}}{{{{(n!)}^2}}}\gamma \left( {1+n, \frac{w_1 d_{\mathrm{s}}^{\alpha} x}{2b}} \right)},
\end{align}
where $w_1 = 16 \pi^2 f_\mathrm{c}^2 N_0 W /(c^2 P G_{\mathrm{t}} G_{\mathrm{r,ml}}  )$ and $\ds=\sqrt{r^2 \sin^2\thes + \as^2 + 2r\as}-r\sin\thes$.
\end{lem}

\begin{proof}[Proof:\nopunct]
The CDF of the SNR at the serving satellite is given by 
\begin{align}\label{eq:CDF_SNRs}
\cdfgs 
    &= \mathbb{P}[\gs<x]
    =\mathbb{P}\left[\hs < \frac{N_0 W x}{P \ell(d_{\mathrm{s}})}\right]
    =F_\hs (w_1 d_{\mathrm{s}}^{\alpha} x).
\end{align}
From \eqref{eq:CDF_ch_gain} and \eqref{eq:CDF_SNRs}, the CDF of the SNR at  the serving satellite is obtained.
\end{proof}

\begin{rem}\label{Rem:der_thes}
It is expected that the SNR at the serving satellite with the higher elevation angle is more likely to be larger than that with the lower elevation angle. To show this, with the fact that the derivative of the lower incomplete gamma function is given by $\frac{d\gamma(a,x)}{dx}=e^{-x}x^{a-1}$,
the derivative of $\cdfgs$ with respect to $\thes$ is obtained as
\begin{align}
\frac{d\cdfgs}{d\thes}
    = \alpha K e^{-\frac{w_1 \ds^\alpha x}{2b}}\sum\limits_{n = 0}^\infty \frac{{{(m)}_n}{\delta ^n}{{(w_1 x)}^{1 + n} }}{{{{(n!)}^2\ds^{1-\alpha (1+n)}}}} \frac{d}{d\thes} \ds,
\end{align}
where
\begin{align}\label{eq:der_ds}
\frac{d}{d\thes}\ds
    &=r \cos\thes \left(\frac{r \sin\thes}{\sqrt{r^2 \sin^2\thes + \as^2 + 2r\as}}-1\right).
\end{align}
Since $r \sin\thes < \sqrt{r^2 \sin^2\thes + \as^2 + 2r\as}$,  $\frac{d\cdfgs}{d\thes}<0$ for all $x$, i.e., $\cdfgs$ is a monotonically decreasing function of $\thes$.
Note that, when the CDF of a random variable $X_1$ is less than that of another random variable $X_2$, i.e., $F_{X_1}(x) < F_{X_2}(x), \forall{x}$, $X_1$ is more likely to be larger than $X_2$.
\end{rem}

\begin{rem}
Similarly, the derivative of $\cdfgs$ with respect to $\as$ is given by 
\begin{align}
\frac{d\cdfgs}{d\as}
    \!=\!  \frac{\alpha K e^{-\frac{w_1 \ds^\alpha x}{2b}}(\as+r)}{\sqrt{r^2 \sin^2\thes + \as^2 + 2r\as}} \sum\limits_{n = 0}^\infty  \frac{{{{(m)}_n}{\delta ^n}{{(w_1 x)}^{1 + n}}}}{{{{(n!)}^2 \ds^{1-\alpha(1+n)}}}}.
\end{align}
Since $\frac{d\cdfgs}{d\as}>0$, it is expected that the SNR at the serving satellite with the lower altitude is more likely to be larger than that with the higher altitude.
\end{rem}

We denote $\gemdml$ and $\gemdsl$ as the SNRs at the most detrimental eavesdropping satellite in $\Aeml$ and $\Aesl$, respectively, when the numbers of eavesdropping satellites in $\Aeml$ and $\Aesl$ are $p$ and $q$, i.e., $\BPP(\Aeml)=p$ and $\BPP(\Aesl)=q$. Then, the most detrimental eavesdropping satellite in $\Ae$ is given by $\g_{\mathrm{e^*}}^{(p,q)}=\max\{\gemdml, \gemdsl\}$.
With these definitions, the following two lemmas show the CDFs of $\gemdml$ and $\gemdsl$, respectively.

\begin{lem}\label{Lem:CDF_SNReml}
Given the number of eavesdropping satellites in $\Aeml$ is $p$, the CDF of the SNR at the most detrimental eavesdropping satellite in $\Aeml$ is given by 
\begin{align}\label{eq:CDF_SNReml_fin}
{F_{\gemdml}}(x) 
    =& \left[\frac{2K}{\dth^2-\av^2}\sum\limits_{n = 0}^\infty  {\frac{{{{(m)}_n}{\delta ^n}{{(2b)}^{1 + n}}}}{{{{(n!)}^2}}}} \mathcal{I}_{\mathrm{ml}}(x,n)\right]^{p},
\end{align}
where 
\begin{align}
&\mathcal{I}_{\mathrm{ml}}(x,n)
    =\frac{\dth^2}{2} \gamma\left(1+n,\frac{w_1 \dth^{\alpha} x}{2b} \right)-\frac{\av^2}{2}\gamma\left(1+n,\frac{w_1 \av^{\alpha} x}{2b} \right)\nonumber\\
    &\!-\!\frac{1}{2}\!\left(\frac{2b}{w_1 x}\right)^{\frac{2}{\alpha}}\!\! \left\{\gamma\!\left(1+n+\frac{2}{\alpha},\frac{w_1 \dth^{\alpha} x}{2b}\right)-\gamma\!\left(1+n+\frac{2}{\alpha},\frac{w_1 \av^{\alpha} x}{2b}\right)\right\}.  
\end{align}
\end{lem}

\begin{proof}[Proof:\nopunct]
See Appendix \ref{App:Lem8}.
\end{proof}

\begin{lem}\label{Lem:CDF_SNResl}
Given the number of eavesdropping satellites in $\Aesl$ is $q$, the CDF of the SNR at the most detrimental eavesdropping satellite in $\Aesl$ is given by 
\begin{align}\label{eq:CDF_SNResl_fin}
{F_{\gemdsl}}(x) 
    =\left[ \frac{2K}{\dmax^2-\dth^2}\sum\limits_{n = 0}^\infty  {\frac{{{{(m)}_n}{\delta^n}{{(2b)}^{1 + n}}}}{{{{(n!)}^2}}}}\mathcal{I}_{\mathrm{sl}}(x,n) \right]^q,
\end{align}
where 
\begin{align}
&\mathcal{I}_{\mathrm{sl}}(x,n)
    =\frac{\dmax^2}{2} \gamma\left(1+n,\frac{w_2 \dmax^{\alpha} x}{2b}\right)-\frac{\dth^2}{2}\gamma\left(1+n,\frac{w_2 \dth^{\alpha} x}{2b}\right)\nonumber\\
    &\!-\!\frac{1}{2}\!\left(\frac{2b}{w_2 x}\right)^{\frac{2}{\alpha}}\! \left\{\gamma\!\left(1\!+n+\frac{2}{\alpha},\frac{w_2 \dmax^{\alpha} x}{2b}\right)\!-\!\gamma\!\left(1\!+n+\frac{2}{\alpha},\frac{w_2 \dth^{\alpha} x}{2b}\right)\right\}
\end{align}
with $w_2 = 16 \pi^2 f_\mathrm{c}^2 N_0 W /(c^2 P G_{\mathrm{t}} G_{\mathrm{r,sl}}  )$.
\end{lem}

\begin{proof}[Proof:\nopunct]
The proof is similar to that of Lemma \ref{Lem:CDF_SNReml}.
\end{proof}

Using these lemmas, the CDF of $\g_{\mathrm{e^*}}^{(p,q)}$ is obtained in the following lemma.

\begin{lem}\label{Lem:CDF_SNRe*}
Given the numbers of eavesdropping satellites in $\Aeml$ and $\Aesl$ are $p$ and $q$, the CDF of the SNR at the most detrimental eavesdropping satellite in $\Ae$ is the product of the CDFs of $\gemdml$ and $\gemdsl$, i.e.,
\begin{align}\label{eq:CDF_SNRe*_fin}
F_{\g_{\mathrm{e^*}}^{(p,q)}}(x)
    =F_{\gemdml}(x)F_{\gemdsl}(x).
\end{align}
\end{lem}

\begin{proof}[Proof:\nopunct]
The CDF of $\g_{\mathrm{e^*}}^{(p,q)}$ is given by
\begin{align}
F_{\g_{\mathrm{e^*}}^{(p,q)}}&(x)
    = \mathbb{P}[\gemd \le x|\BPP(\Aeml)=p,\BPP(\Aesl)=q]\nonumber\\
    &= \mathbb{E}_{X_\idxe,Y_\idxe}\left[ \prod\limits_{\idxe \in {\BPPeml}} {\mathbb{P}[{\gamma_\idxe} \le x]}
    \prod\limits_{\idxe \in {\BPPesl}} {\mathbb{P}[{\gamma_\idxe} \le x]}
    \right]\nonumber\\
    &\mathop=\limits^{(a)} \mathbb{E}_{X_\idxe}\left[ \prod\limits_{\idxe \in {\BPPeml}} {\mathbb{P}[{\gamma_\idxe} \le x]}\right]\mathbb{E}_{Y_\idxe}\left[
    \prod\limits_{\idxe \in {\BPPesl}} {\mathbb{P}[{\gamma_\idxe} \le x]}
    \right],
\end{align}
where ($a$) follows from the independence between the random variables $X_\idxe$ and $Y_\idxe$.
From the definition of the CDFs of $\gemdml$ and $\gemdsl$, the proof is complete.
\end{proof}

The derived expressions for the CDFs of $\gs$ and $\g_{\mathrm{e^*}}^{(p,q)}$ will be useful to derive the ergodic and outage secrecy capacities in the next subsections.

\subsection{Ergodic Secrecy Capacity}
The instantaneous secrecy rate of the system is defined as the data rate securely transferred from the terminal to the serving satellite, which is given by [\ref{Ref:Wyner}]
\begin{align}\label{eq:R}
R
    =[\log_2(1+\gs)-\log_2(1+\gemd)]^+,
\end{align}
where $[x]^+ = \max\{0,x\}$. 
The ergodic secrecy capacity of the system is the average of the instantaneous secrecy rate, i.e., $\Cerg=\mathbb{E}[R]$, which is obtained in the following theorem.

\begin{thm}\label{thm:Cerg}
The ergodic secrecy capacity of the uplink satellite communication system in the presence of $N$ eavesdropping satellites randomly located at the altitude $\av$ is given by
\begin{align}\label{eq:Cerg_fin}
\Cerg
    \!=\!\sum_{p=0}^{N}\! \sum_{q=0}^{N-p} \frac{\mathcal{P}[N,p,q]}{\ln2}\!\! \int_0^{\infty}\!\frac{F_{\g_{\mathrm{e}^*}^{(p,q)}}(x)}{1+x}(1-F_\gs(x))dx,
\end{align}
where $\mathcal{P}[N,p,q]$, $F_{\gs}(x)$, and $F_{\g_{\mathrm{e}^*}^{(p,q)}}(x)$ are given in \eqref{eq:prob_pq_fin}, \eqref{eq:CDF_SNRs_fin}, and \eqref{eq:CDF_SNRe*_fin}, respectively.
\end{thm} 

\begin{proof}[Proof:\nopunct]
See Appendix \ref{App:Thm1}.
\end{proof}

\begin{rem}
Since the integrands in \eqref{eq:Cerg_fin} decrease with the CDF of $\gs$ and increase with the CDF of $\g_{\mathrm{e}^*}^{(p,q)}$, the ergodic secrecy capacity increases as $\gs$ increases or $\g_{\mathrm{e}^*}^{(p,q)}$ decreases, which is expected from the definition of the secrecy rate \eqref{eq:R}.
\end{rem}

\subsection{Outage Secrecy Capacity}
A secrecy outage occurs when the secrecy rate of the system is less than or equal to a target secrecy rate $\Rt$.
Mathematically, the secrecy outage probability of the system is defined as $\Pout=\mathbb{P}[R\le\Rt]$, which is obtained in the following theorem.

\begin{thm}\label{thm:Pout}
The secrecy outage probability of the uplink satellite communication system in the presence of $N$ eavesdropping satellites randomly located at the altitude $\av$ is given by
\begin{align}\label{eq:Pout_fin}
\Pout
    =&1-\sum_{p=0}^{N} \sum_{q=0}^{N-p} \mathcal{P}[N,p,q]\nonumber\\
    &\qquad\times \int_{2^{\Rt}-1}^{\infty} F_{\g_{\mathrm{e}^*}^{(p,q)}}(2^{-\Rt}(1+x)-1) f_{\gs}(x)dx,
\end{align}
where $\mathcal{P}[N,p,q]$ and $F_{\g_{\mathrm{e}^*}^{(p,q)}}(x)$ are given in \eqref{eq:prob_pq_fin} and \eqref{eq:CDF_SNRe*_fin}, respectively, and 
$\pdfgs$ is given by
\begin{align}\label{eq:PDF_SNRs}
\pdfgs 
    &= K e^{- w_1 d_{\mathrm{s}}^{\alpha } x} \sum\limits_{n = 0}^\infty  \frac{{{(m)}_n (\delta x)^n (2b w_1)^{1+n} }}{(n!)^2 d_{\mathrm{s}}^{-\alpha (1+n)}} 
\end{align}
\end{thm} 

\begin{proof}[Proof:\nopunct]
See Appendix \ref{App:Thm2}.
\end{proof}

The outage secrecy capacity is defined as the maximum secrecy rate successfully transferred from the terminal to the serving satellite with a target secrecy outage probability $\epsilon$. Using Theorem \ref{thm:Pout}, the outage secrecy capacity is given in the following corollary.

\begin{cor}\label{cor:Cout}
With a target secrecy outage probability $\epsilon$, the outage secrecy capacity of the uplink satellite communication system in the presence of $N$ eavesdropping satellites randomly located at the altitude $\av$ is given by $\Cout=(1-\epsilon)\Rt^*$, where $\Rt^*$ is the target secrecy rate satisfying $\Pout=\epsilon$.
\end{cor} 

The closed-form expressions of \eqref{eq:Cerg_fin} and \eqref{eq:Pout_fin} are difficult to obtain due to the integrals, which can instead be computed numerically.
In addition, the optimal target secrecy rate $\Rt^*$ for the outage secrecy capacity can be easily obtained using a simple line-search since $\Pout$ is an increasing function of $\Rt$.

\subsection{Extension to Steerable-Beam Antennas With Terminal's Location Information}\label{Sec:Secrecy_performance_anal_sb}
So far, we considered the eavesdropping satellites with the fixed-beam antennas.
In this subsection, we assume that the eavesdropping satellites are equipped with the steerable-beam antennas and have information of the terminal's position, so that the eavesdropping satellites are able to steer the main lobes of the beams in the direction of the terminal.
Let $\wsb$ indicate the steerable angle of the boresight from the direction to the subsatellite point. 
Then, the steerable-beam case can be seen as the case that the beamwidth of the main lobe is increased by $\wsb$ without any loss of the antenna gain, i.e., $\tilde{\omega}_{\mathrm{th}}={\omega}_{\mathrm{th}} + \wsb$. 
For example, when $\wsb$ is close to zero, the case becomes comparable to that with the fixed-beam antennas, while when $\wsb$ is large enough, the main lobes of most of effective eavesdropping satellites can be directed to the terminal.
The analytical results for the case with the steerable-beam antennas can be easily obtained through the same steps for the fixed-beam case  with the different threshold angle $\tilde{\omega}_{\mathrm{th}}$.

\begin{rem}\label{Rem:wsb*}
It can be seen from Fig. \ref{Fig:area_description} that the main lobes of the satellites at the boundary of $\Ae$ are directed to the terminal if $\sin\wth\geq r/(r+\av)$. 
With the steerable-beam antennas, the condition that all effective eavesdropping satellites can steer the main lobes to the terminal is given by $ \sin(\wth+\wsb) \geq r/(r+\av)$.
Hence, the minimum steerable angle with which all effective eavesdropping satellites can steer the main lobe to the terminal, is given by
\begin{align}\label{eq:wsb*}
\Delta \omega_{\mathrm{sb}}^*
    =\sin^{-1}\left(\frac{r}{r+\av}\right)-\wth.
\end{align}
As $\Delta \omega_{\mathrm{sb}}$ increases, the surface area of $\Aeml$ is more expanded and finally becomes that of $\Ae$ when $\Delta \omega_{\mathrm{sb}}$ reaches $\Delta \omega_{\mathrm{sb}}^*$.
\end{rem}

\section{Approximation and Asymptotic Analyses for Secrecy Performance}\label{Sec:Asym_secrecy_performance}
The expressions for the secrecy performance in Theorems \ref{thm:Cerg} and \ref{thm:Pout} are exact but complicated to evaluate because $\frac{(N+1)(N+2)}{2}$ integrals need to be computed numerically.
In this section, we approximate the secrecy performance to reduce the computational complexity, assuming that the altitude of the eavesdropping satellites is sufficiently low.\footnote{This assumption is valid for the eavesdropping satellites with the low altitudes, e.g., LEO and very low Earth orbit (VLEO) satellites, and its validity will be shown in Figs. \ref{Fig:Cerg_vs_ae} and \ref{Fig:Pout_vs_ae}.}
In addition, simpler expressions are derived to investigate asymptotic behavior of the secrecy performance.

\subsection{Performance Approximation}\label{Sec:Perf_approx}
From the definition of the BPP, the number of eavesdropping satellites in $\Ae$ follows the binomial distribution with the total number of points $N$ and the success probability\footnote{Note that the success probability is the probability that a point is located on the area of interest. For homogeneous BPPs, the success probability is obtained as the ratio of the surface area of interest to the total surface area where all points are distributed [\ref{Ref:Book:Chiu}].} $p(\Ae)=\SAe/\SA=\frac{\av}{2(r+\av)}$.
Based on the Poisson limit theorem, when the success probability of the binomial distribution is very small, the distribution asymptotically follows the Poisson distribution [\ref{Ref:Book:Chiu}].
Thus, when $\av \rightarrow 0$, $p(\Ae)\rightarrow 0$ so that the BPP in $\Ae$, $\BPPe$, can be approximated as a PPP $\PPPe$ with the density of $\lambda_{\mathrm{e}}=\frac{N}{4\pi(r+\av)^2}$.
% $\lambda_{\mathrm{e}}=N/(4\pi(r+\av)^2)$.
The PPP $\PPPe$ can be divided into two PPPs $\PPPeml$ and $\PPPesl$ for the two disjoint regions $\Aeml$ and $\Aesl$, respectively. 

We denote $\gemdmlapp$ and $\gemdslapp$ as the SNRs at the most detrimental eavesdropping satellite in $\Aeml$ and $\Aesl$, respectively, when $\av \rightarrow 0$. Then, the CDFs of $\gemdmlapp$ and $\gemdslapp$ are given in the following two lemmas.

\begin{lem}\label{Lem:CDFgemlapp}
When $\av \rightarrow 0$, the CDF of the SNR at the most detrimental eavesdropping satellite in $\Aeml$, $\gemdmlapp$, is approximated as
\begin{align}\label{eq:CDF_SNRemlapp_fin}
&{F_{\gemdmlapp}}(x) =\exp \left( { - \frac{N}{2}\left\{ {1 - \cos \psth - K\sum\limits_{n = 0}^\infty  {\frac{{{{(m)}_n}{\delta ^n}{{(2b)}^{1 + n}}}}{{{{(n!)}^2}}}} } \right.} \right.\nonumber\\
    &\times\left[\left\{ {\frac{{{r^2} + {{(r + {a_{\mathrm{e}}})}^2}}}{{2r(r + {a_{\mathrm{e}}})}} - \cos \psth} \right\} \gamma (1 + n,\Pi_{\mathrm{ml}}(x,\psth))  \right.\nonumber\\
    &- \frac{{a_{\mathrm{e}}^2}}{{2r(r + {a_{\mathrm{e}}})}}\gamma (1 + n,\Pi_{\mathrm{ml}}(x,0)) - \frac{1}{{2r(r + {a_{\mathrm{e}}})(w_1 x)^{\frac{2}{\alpha}}}}\nonumber\\
    &\left. {\left. {\left. \!\!\times \!\left\{ {\gamma \left( 1 + n + \frac{2}{\alpha },\Pi_{\mathrm{ml}}(x,\psth) \right) - \gamma \left( 1 + n + \frac{2}{\alpha },\Pi_{\mathrm{ml}}(x,0) \right)} \right\} \right]} \right\}} \right),
\end{align}
where 
\begin{align}\label{eq:Pi_ml}
\Pi_{\mathrm{ml}}(x,\psi)
    =\frac{w_1 x}{2b}\left(\!\sqrt{r^2+(r+\av)^2-2r(r+\av)\cos\psi}\right)^{\alpha}\!.
\end{align}
\end{lem}
\begin{proof}[Proof:\nopunct]
See Appendix \ref{App:Cor3}.
\end{proof}

\begin{lem}\label{Lem:CDFgeslapp}
When $\av \rightarrow 0$, the CDF of the SNR at the most detrimental eavesdropping satellite in $\Aesl$, $\gemdslapp$, is approximated as
\begin{align}\label{eq:CDF_SNReslapp_fin}
&{F_{\gemdslapp}}(x) 
    = \!\exp \left(  - \frac{N}{2}\!\left\{ \cos \psth - \frac{r}{r+\av} - K\!\sum\limits_{n = 0}^\infty \! \frac{{{{(m)}_n}{\delta ^n}{{(2b)}^{1 + n}}}}{{{{(n!)}^2}}}\right. \right.\nonumber\\
    &\!\times \!\left[ \frac{\av(\av+2r)}{2r(r+\av)}\left\{ {\gamma \left( 1 + n,\Pi_{\mathrm{sl}}(x,\psmax) \right)} \!-\! \gamma \left( 1 + n,\Pi_{\mathrm{sl}}(x,\psth) \right) \right\}\right.\nonumber\\
    &\!+ {\left(\cos \psth - \frac{r}{r+\av}\right)} \gamma \left( 1 + n, \Pi_{\mathrm{sl}}(x,\psth) \right)- \frac{1}{{2r(r + {a_{\mathrm{e}}})(w_2 x)^{\frac{2}{\alpha}}}}\nonumber\\
    &\!\! \times \!\left\{ \gamma \!\left( 1 + n + \frac{2}{\alpha },\Pi_{\mathrm{sl}}(x,\psmax)  \right) \!- \!\!\!\left. \left. \left.  \gamma \!\left( 1 + n + \frac{2}{\alpha },\Pi_{\mathrm{sl}}(x,\psth) \right) \right\} \right] \right\} \right),
\end{align}
where 
\begin{align}
    \Pi_{\mathrm{sl}}(x,\psi)=\frac{w_2 x}{2b} \left(\sqrt{r^2+(r+\av)^2-2r(r+\av)\cos\psi}\right)^{\alpha}.
\end{align}
\end{lem}

\begin{proof}[Proof:\nopunct]
The proof is similar to that of Lemma \ref{Lem:CDFgemlapp}.
\end{proof}

Let $\gemdapp$ denote the SNR at the most detrimental eavesdropping satellite in $\Ae$ when $\av \rightarrow 0$. Then, the CDF of $\gemdapp$ is given by $F_{\gemdapp}(x)=F_{\gemdmlapp}(x)F_{\gemdslapp}(x)$. 

With this CDF, we obtain the approximated ergodic secrecy capacity and secrecy outage probability in the following theorem.
\begin{thm}\label{thm:asymCerg}
When the altitude of the eavesdropping satellites is low, e.g., LEO or VLEO satellites, the ergodic secrecy capacity and the secrecy outage probability are respectively approximated as
\begin{align}\label{eq:Cerg_assym_fin}
\Cerg
    \approx \frac{1}{\ln2} \int_0^{\infty}\frac{F_{\gemdapp}(x)}{1+x}(1-F_\gs(x))dx
\end{align}
and
\begin{align}\label{eq:Pout_assym_fin}
\Pout
    \approx 1- \int_{2^{\Rt}-1}^{\infty} F_{\gemdapp}(2^{-\Rt}(1+x)-1) f_{\gs}(x)dx.
\end{align}
\end{thm}
\begin{proof}[Proof:\nopunct]
With the expression for the CDF of $\gemdapp$, the proof is complete with the similar steps as in the proofs of Theorems~\ref{thm:Cerg} and \ref{thm:Pout}.
\end{proof}

Compared to the exact expressions derived in \eqref{eq:Cerg_fin} and \eqref{eq:Pout_fin},
the simplified expressions \eqref{eq:Cerg_assym_fin} and \eqref{eq:Pout_assym_fin} significantly reduce the computational complexity to evaluate the performance.
Specifically, the computational complexity to evaluate the CDFs of $\gemdml$ and $\gemdsl$ in \eqref{eq:CDF_SNReml_fin} and \eqref{eq:CDF_SNResl_fin} is given by
$O(\tau^p)$ and $O(\tau^q)$, respectively,
where $O(\tau)$ is the computational complexity for the lower incomplete gamma function.
In contrast, the computations of the CDFs of $\gemdmlapp$ and $\gemdmlapp$ only require the complexity of $O(\tau)$ thanks to the probability generating functional of the PPP.
In addition, the approximate expressions do not have $\mathcal{P}[N,p,q]$ requiring the complexity of $O(Np^2q^2(N-p-q)^2)$.
More importantly, $\frac{(N+1)(N+2)}{2}$ integrals in the exact expressions are reduced to only one in the approximate ones, which is the key for the complexity reduction.
For example, when $N=10,000$, around $50$ million integrals should be calculated in the exact expressions, while only a single integral is required in the approximate ones.

% as, thes, N, ae, wth, deltawsb, epsilon(or Rt)
\begin{figure}
\begin{center}
\includegraphics[width=0.90\columnwidth]{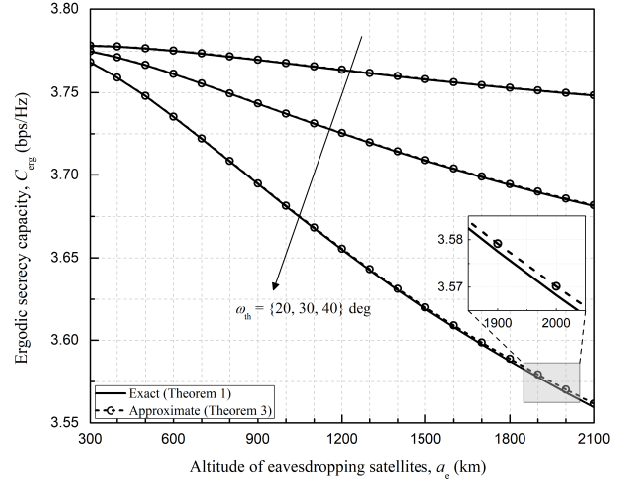}
\end{center}
\setlength\abovecaptionskip{.25ex plus .125ex minus .125ex}
\setlength\belowcaptionskip{.25ex plus .125ex minus .125ex}
\caption{Ergodic secrecy capacity versus the altitude of the eavesdropping satellites $\av$ for various threshold angles $\wth=\{20,30,40\}$ deg with $N=10$, $\as=\av=600$, $\thes=60$ deg, and $\wsb=0$ deg.}

\label{Fig:Cerg_vs_ae}
\end{figure}

% as, thes, N, ae, wth, deltawsb, epsilon(or Rt)
\begin{figure}
\begin{center}
\includegraphics[width=0.90\columnwidth]{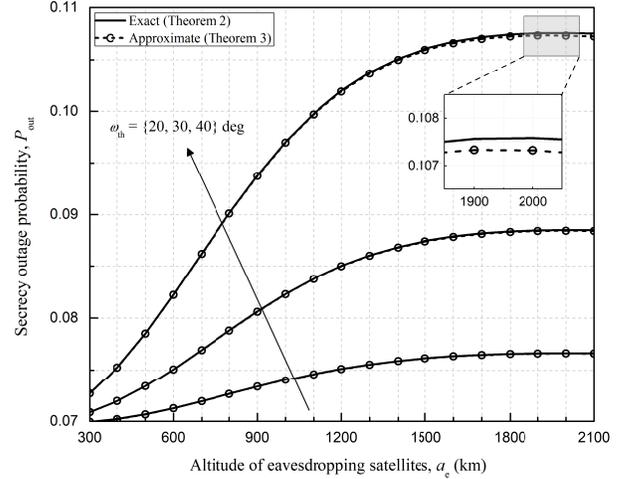}
\end{center}
\setlength\abovecaptionskip{.25ex plus .125ex minus .125ex}
\setlength\belowcaptionskip{.25ex plus .125ex minus .125ex}
\caption{Secrecy outage probability versus the altitude of the eavesdropping satellites $\av$ for various threshold angles $\wth=\{20,30,40\}$ deg with $N=10$, $\as=\av=600$, $\thes=60$ deg, $\wsb=0$ deg, and $\Rt=2$ bps/Hz.}

\label{Fig:Pout_vs_ae}
\end{figure}

In Figs. \ref{Fig:Cerg_vs_ae} and \ref{Fig:Pout_vs_ae}, we compare the exact expressions of the secrecy performance to the approximate ones.
For the altitudes from $300$ to $2,100$ km, at which the LEO and VLEO satellites are usually located, the secrecy performance is degraded with the altitude of the eavesdropping satellites $\av$ for a given~$N$. This shows that the eavesdroppers better overhear desired signals with the higher altitude, which leads to the larger coverage for a given 3-dB beamwidth.
As we expect from the assumption $\av \rightarrow 0$ for the approximate analyses, as $\av$ decreases, the gap between the exact and approximate performance decreases.
These results show the validity of the Poisson limit theorem, i.e., when the success probability is sufficiently small, the binomial distribution asymptotically becomes the Poisson distribution.
Thus, the approximate expressions can be used to evaluate the secrecy performance with low complexity but high accuracy.

\begin{figure*}[b]
\setcounter{myeqncount}{\value{equation}}
\setcounter{equation}{39}
\normalsize \hrulefill \vspace*{4pt}
\begin{align}\label{eq:Lamxy_fin}
\Lambda_{x,y}= \frac{N}{2r(r+\av)\ds}\left\{F_{d_{\idxe}}(x)^N-F_{d_{\idxe}}(y)^N\right\}^{-1}\sum\limits_{i=0}^{N-1}\binom{N-1}{i}\left(1+\frac{\av^2}{4r(r+\av)}\right)^{N-1-i}\left(\frac{-1}{4r(r+\av)}\right)^i\frac{x^{2i+3}-y^{2i+3}}{2i+3}
\end{align}
\normalsize \hrulefill \vspace*{4pt}
\setcounter{equation}{42}
\begin{align}\label{eq:Linf_fin}
\mathcal{L}_{\infty}
    &= \lim_{P\rightarrow\infty}\left(\log_2{P}-\frac{\Cerg^{\infty}}{S_{\infty}}\right) \nonumber\\
    &= -\frac{\log_2\Lambda_{\dth,\av}}{\left(1-F_{d_{\idxe}}(\dmax)\right)^N} - \frac{\left(1-F_{d_{\idxe}}(\dth)\right)^N}{\left(1-F_{d_{\idxe}}(\dmax)\right)^N}\log_2\left(\frac{\Grml\Lambda_{\dmax,\dth}^{\alpha}}{\Grsl\Lambda_{\dth,\av}^{\alpha}}\right) - \log_2 \left(\frac{c^2 G_{\mathrm{t}}\Grsl}{16 \pi^2 f_{\mathrm{c}}^2 N_0 W \Grml \ds^{\alpha}\Lambda_{\dmax,\dth}^{\alpha}}\right)
\end{align}
\setcounter{equation}{\value{myeqncount}}
\end{figure*}

\subsection{Asymptotic Analyses}\label{Sec:Asym_real}
In this subsection, we obtain simpler expressions to investigate asymptotic behavior of the secrecy performance in the three different scenarios.

% as, thes, N, ae, deltawsb, epsilon(or Rt)
\begin{figure}
\begin{center}
\includegraphics[width=0.90\columnwidth]{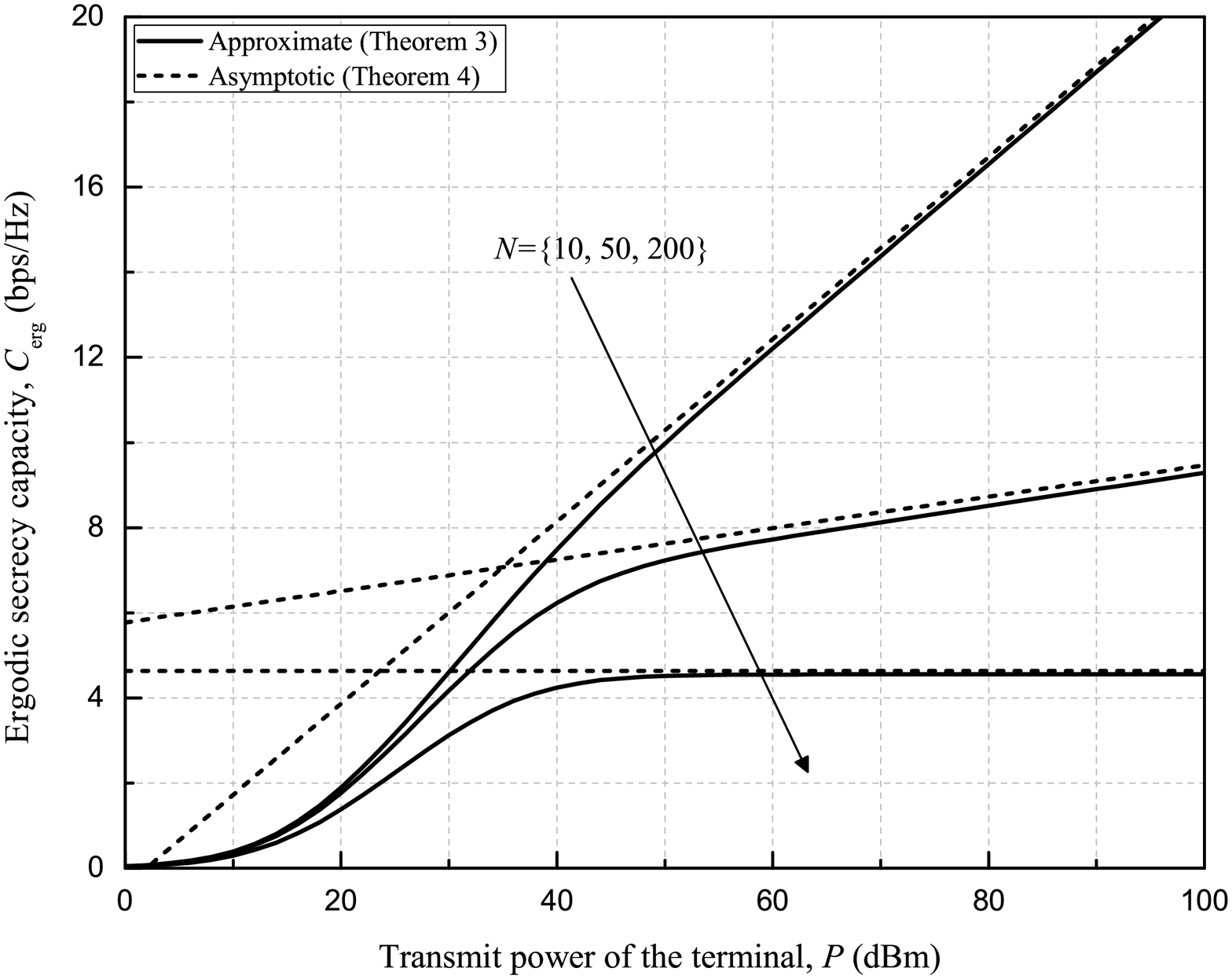}
\end{center}
\setlength\abovecaptionskip{.25ex plus .125ex minus .125ex}
\setlength\belowcaptionskip{.25ex plus .125ex minus .125ex}
\caption{Ergodic  secrecy capacity versus the transmit power of the terminal $P$ for various numbers of the eavesdropping satellites $N=\{10,50,200\}$ with $\as=\av=600$ km, $\thes=60$ deg, $\wth=40$ deg, and $\wsb=0$ deg.}

\label{Fig:Cerg_vs_P}
\end{figure}

\subsubsection{When there is no malicious eavesdropping satellite, i.e., $N\rightarrow0$} This scenario describes the situation where all satellites are trustworthy, i.e., the satellites do not eavesdrop on  terminal's signals. In such a case, the asymptotic ergodic secrecy capacity and secrecy outage probability are obtained in the following corollary.
\begin{cor}\label{cor:Cerg_N0}
When $N\rightarrow0$, the ergodic secrecy capacity and the secrecy outage probability in Theorem \ref{thm:asymCerg} are simplified as
\begin{align}\label{eq:Cerg_N0}
    \Cerg 
    &\rightarrow \frac{K}{\ln{2}}\sum_{k=0}^{\floor{m}-1}\sum_{t=0}^{k}\frac{(-1)^k (1-m)_{k} \delta^k (w_1 \ds^{\alpha})^t}{k! \left(\frac{1}{2b}-\delta\right)^{k-t+1}} \nonumber\\
    &\quad\qt\exp\left(\left(\frac{1}{2b}-\delta\right)w_1 \ds^{\alpha}\right)\Gamma\left(-t,\frac{1}{2b}-\delta\right)
\end{align}
and
\begin{align}\label{eq:Pout_N0}
    \Pout \rightarrow F_{\gs}(2^{\Rt}-1).
\end{align}
\end{cor}
\begin{proof}[Proof:\nopunct]
See Appendix \ref{App:CorCerg_N0}.
\end{proof}

\begin{rem}
The asymptotic expression \eqref{eq:Cerg_N0} gives a theoretical upper-bound of the ergodic secrecy capacity, and \eqref{eq:Pout_N0} is a lower-bound of the secrecy outage probability.
\end{rem}

\subsubsection{When there are a large number of eavesdropping satellites, i.e., $N\rightarrow\infty$}
In this scenario, the secrecy performance with a large number of eavesdroppers is given as follows.
\begin{cor}
When $N\rightarrow \infty$, we have $\Cerg \rightarrow 0$ and $\Pout \rightarrow 1$.
\end{cor}
\begin{proof}[Proof:\nopunct]
As $N$ goes to infinity, the CDF of $\gemdapp$ becomes zero for a given $x$, i.e., $F_{\gemdapp}(x) \to 0$. By letting $F_{\gemdapp}(x)=0$ in \eqref{eq:Cerg_assym_fin} and \eqref{eq:Pout_assym_fin}, the proof is complete.
\end{proof}
\begin{rem}
$F_{\gemdapp}(x) \to 0$ means that the SNR at the most detrimental eavesdropper is certainly very high, e.g., $\gemdapp \rightarrow \infty$. This explains that when there is a large number of eavesdropping satellites, the most detrimental eavesdropping satellite is likely to have a very high eavesdropping capability, which is intuitively true.
\end{rem}

\subsubsection{When the transmit power of the terminal is large, i.e., $P\rightarrow\infty$}
In the high-SNR regime, the asymptotic ergodic secrecy capacity is obtained in the following theorem.
\begin{thm}\label{thm:CergHighSNR}
In the high-SNR regime, the ergodic secrecy capacity is asymptotically upper-bounded as
\begin{align}\label{eq:Cerg_Pinf}
    \lim_{P\rightarrow\infty}{\Cerg(P)} \le  \Cerg^{\infty},
\end{align}
where 
\begin{align}\label{eq:Cerginf_fin}
\Cerg^{\infty}
    &=\log_2\Lambda_{\dth,\av} + \left(1-F_{d_{\idxe}}(\dth)\right)^N\log_2\left(\frac{\Grml\Lambda_{\dmax,\dth}^{\alpha}}{\Grsl\Lambda_{\dth,\av}^{\alpha}}\right)\nonumber\\
    &\quad +\left(1-F_{d_{\idxe}}(\dmax)\right)^N\log_2 \left(\frac{\Grsl}{\Grml w_1 \ds^{\alpha}\Lambda_{\dmax,\dth}^{\alpha}}\right)
\end{align}
with $\Lambda_{x,y}$ given in \eqref{eq:Lamxy_fin} at the bottom of this page.
\end{thm}
\begin{proof}[Proof:\nopunct]
See Appendix \ref{App:CorCergHighSNR}.
\end{proof}

\setcounter{equation}{40}
The asymptotic upper-bound of the ergodic secrecy capacity in \eqref{eq:Cerginf_fin} can be expressed as [\ref{Ref:Jin}]
\begin{align}
\Cerg^{\infty}
    = S_{\infty}(\log_2 P - \mathcal{L}_{\infty}),
\end{align}
where $S_{\infty}$ denotes the high-SNR slope in bps/Hz (3 dB), and $\mathcal{L}_{\infty}$ is the high-SNR power offset in 3 dB units. After some manipulations, the high-SNR slope of $\eqref{eq:Cerginf_fin}$ is obtained as
\begin{align}\label{eq:HS_slope}
S_{\infty}
    = \lim_{P\rightarrow\infty}\frac{\Cerg^{\infty}}{\log_2{P}}
    =\left(1-\frac{0.5}{1+r/\av}\right)^N,
\end{align}
and the high-SNR power offset is given in \eqref{eq:Linf_fin} at the bottom of this page.
\begin{rem}\label{rem:CergHighSNR}
In \eqref{eq:HS_slope}, as $r/\av$ is always positive, $1-\frac{0.5}{1+r/\av}$ lies in $(0.5,1)$. From this observation, we can know that the high-SNR slope $S_{\infty}$  decreases with the number of eavesdroppers~$N$.
\end{rem}
\setcounter{equation}{43}

Fig. \ref{Fig:Cerg_vs_P} shows the ergodic secrecy capacity versus the transmit power of the terminal $P$ for various numbers of eavesdropping satellites. This figure verifies Theorem \ref{thm:CergHighSNR}, showing that the gap between the asymptotic and approximate results becomes smaller as $P$ increases. As mentioned in Remark~\ref{rem:CergHighSNR}, the asymptotic slope gets smaller as the number of the eavesdroppers increases, which means a greater eavesdropping capability.

\section{Eavesdropping Satellites at Different Altitudes}\label{sec:sat_diff_alti}
In this section, we extend the analytical results to the general case that the satellites are located at different altitudes, which can be applied to practical LEO satellite constellations.
Assume that $\Nl$, $v=\{1,2,\cdots,V\}$, eavesdropping satellites are located at the altitude $\al$ according to a BPP $\BPPl$, where $\sum_{v=1}^{V}\Nl=N$.
Then, the SNR at the most detrimental eavesdropping satellite at the altitude $\al$ is given by
\begin{align}\label{eq:SNRe*l}
\glemd
    =& \begin{cases} 
    \max\limits_{\idxe\in\BPPle} \g_\idxe, & \mbox{if } \BPPle \neq \emptyset,\\
    0, & \mbox{if } \BPPle = \emptyset,
     \end{cases}
\end{align}
where $\BPPle$ is the set of effective eavesdropping satellites at the altitude $\al$. The SNR at the most detrimental eavesdropping satellite among all the satellites is given by 
\begin{align}\label{eq:SNRe*L}
\gLemd
    = \max\limits_{v\in\{1,2,\cdots,V\}} \glemd.
\end{align}
Using Lemmas \ref{Lem:CDFgemlapp} and \ref{Lem:CDFgeslapp}, the CDF of $\gLemd$ is given in the following lemma.

\begin{lem}\label{Lem:CDF_SNRe*L}
When $\Nl$ eavesdropping satellites are randomly located at the altitude $\al$, $v=\{1,2,\cdots,V\}$, the CDF of the SNR at the most detrimental eavesdropping satellite is approximated as
\begin{align}\label{eq:CDF_SNRe*L}
F_{\gLemd}(x)
    \approx \prod_{v=1}^V F_{\gemdmlapp}(x|\Nl,\al)F_{\gemdslapp}(x|\Nl,\al),
\end{align}
where $F_{\gemdmlapp}(x|\Nl,\al)$ and $F_{\gemdslapp}(x|\Nl,\al)$ are the CDFs of $\gemdmlapp$ and $\gemdslapp$ in Lemmas \ref{Lem:CDFgemlapp} and \ref{Lem:CDFgeslapp} with $N=\Nl$ and $\av=\al$.
\end{lem}

\begin{proof}[Proof:\nopunct]
The CDF of $\gLemd$ is given by $F_{\gLemd}(x)=\P[\gLemd \le x] = \prod_{v=1}^V\P[\glemd \le x]= \prod_{v=1}^V F_{\glemd}(x)$.
The CDF of $\glemd$ is approximated as $F_{\glemd}(x) \approx F_{\gemdmlapp}(x|\Nl,\al)F_{\gemdslapp}(x|\Nl,\al)$ by using Lemmas \ref{Lem:CDFgemlapp} and \ref{Lem:CDFgeslapp}, which completes the proof. 
\end{proof}

This lemma is used to derive the ergodic secrecy capacity and secrecy outage probability for the satellites with the different altitudes in the following theorem.

\begin{thm}\label{thm:CergL}
The ergodic secrecy capacity and the secrecy outage probability of the uplink satellite system in the presence of $\Nl$ eavesdropping satellites randomly located at the altitude $\al$, $v=\{1,2,\cdots,V\}$, are respectively approximated as
\begin{align}\label{eq:CergL}
\CergL
    \approx \frac{1}{\ln2} \int_0^{\infty}\frac{F_{\gLemd}(x)}{1+x}(1-F_\gs(x))dx
\end{align}
and
\begin{align}
\PoutL 
    \approx 1-\int_{2^{\Rt}-1}^{\infty} F_{\gLemd}(2^{-\Rt}(1+x)-1) f_{\gs}(x)dx.
\end{align}
\end{thm}

\begin{proof}[Proof:\nopunct]
Using the derived CDF of $\gLemd$ in \eqref{eq:CDF_SNRe*L}, this result can be obtained with the similar steps as in the proof of Theorem~\ref{thm:asymCerg}.
\end{proof}

This theorem provides the approximate results that can be applicable to the satellites with arbitrary low altitudes and becomes the result of Theorem \ref{thm:asymCerg} when a single altitude is considered.

\section{Simulation Results}\label{Sec:Sim_results}

In this section, we numerically verify the derived results with the simulation parameters listed in Table \ref{Table:Sim_Param} unless otherwise stated, where the parameters for the terminal are based on the handheld-type terminal considered in the 3GPP NTN standard [\ref{Ref:3GPP_38.821}].
For the shadowed-Rician fading, the average shadowing is assumed with $\{b=0.126,\ m=10.1,\ \Omega=0.835\}$ [\ref{Ref:Abdi}].
Except for the simulations shown in Fig. \ref{Fig:C_vs_as}, we set the number of altitudes for the eavesdropping satellites to $V=1$ and denote the altitude by $\av$ for simplicity of notation.
For the single altitude constellation, i.e., $V=1$, the analytical results for the secrecy performance are obtained from Theorems \ref{thm:Cerg} and \ref{thm:Pout}, 
while for the multi-altitude constellation, i.e., $V>1$, the analytical results are given from Theorem \ref{thm:CergL}.

\begin{table}[t]
    \caption{Simulation Parameters}\label{Table:Sim_Param}
    \centering
    \begin{tabular}{|l|c|}
     \hline 
     Parameter & Value \\
     \hline\hline
     Radius of the Earth $r$ [km] & 6,378\\
     \hline
     Path-loss exponent $\alpha$ & 2\\
     \hline
     Speed of light $c$ [m/s] & $3 \times 10^8$\\
     \hline
     Carrier frequency $f_{\mathrm{c}}$ [GHz] & 2 \\
     \hline
     Carrier bandwidth $W$ [MHz] & 1\\
     \hline
     Transmit power $P$ [dBm] & 23\\
     \hline
     Noise spectral density $N_0$ [dBm/Hz] & $-174$
     \\
     \hline
     Transmit antenna gain $G_{\mathrm{t}}$ [dBi] & 0\\
     \hline
     Receive antenna gain for the main lobe $G_{\mathrm{r,ml}}$ [dBi] & 30\\
     \hline
     Receive antenna gain for the side lobe $G_{\mathrm{r,sl}}$ [dBi] & 10
     \\
     \hline
    \end{tabular}
\end{table}

Fig. \ref{Fig:Pcase_vs_N} shows the probability of distribution cases for the eavesdropping satellites with both fixed-beam (FB) and steerable-beam (SB) antennas. 
When the number of eavesdropping satellites is small, the probability of Case 1, i.e., the probability that there is no effective eavesdropping satellite, is larger than those of the other cases, while as $N$ increases, Case 4 becomes the most probable case. 
This is because, with the large number of eavesdropping satellites, it is more likely for both $\Aeml$ and $\Aesl$ to include at least one satellite.
It is also shown that the probability of Case 2 is much larger than that of Case 3 for the eavesdropping satellites with the FB antennas.
This is because the surface area of $\Aesl$ is $\mathcal{S}_{\Aesl}=51.88 \times 10^6$ $\mathrm{km}^2$, which is approximately 10 times larger than that of $\Aeml$, $\mathcal{S}_{\Aeml}=5.26 \times 10^6$ $\mathrm{km}^2$, so that Case 3 is much less likely to happen than Case 2.
On the other hand, for the eavesdropping satellites with the SB antennas, the surface areas of $\Aeml$ and $\Aesl$ are given by $\mathcal{S}_{\Aeml}=2.56 \times 10^7$ $\mathrm{km}^2$ and $\mathcal{S}_{\Aesl}=3.15 \times 10^7$ $\mathrm{km}^2$, respectively.
The gap between $\mathcal{S}_{\Aeml}$ and $\mathcal{S}_{\Aesl}$ becomes smaller compared to that with the FB antennas, which in turn makes the probabilities of Cases 2 and 3 close to each other.

% as, thes, N, ae, wth, deltawsb, epsilon(or Rt)
\begin{figure}
\centering
\subfigure[Fixed-beam case.]{
\includegraphics[width=0.45\columnwidth]{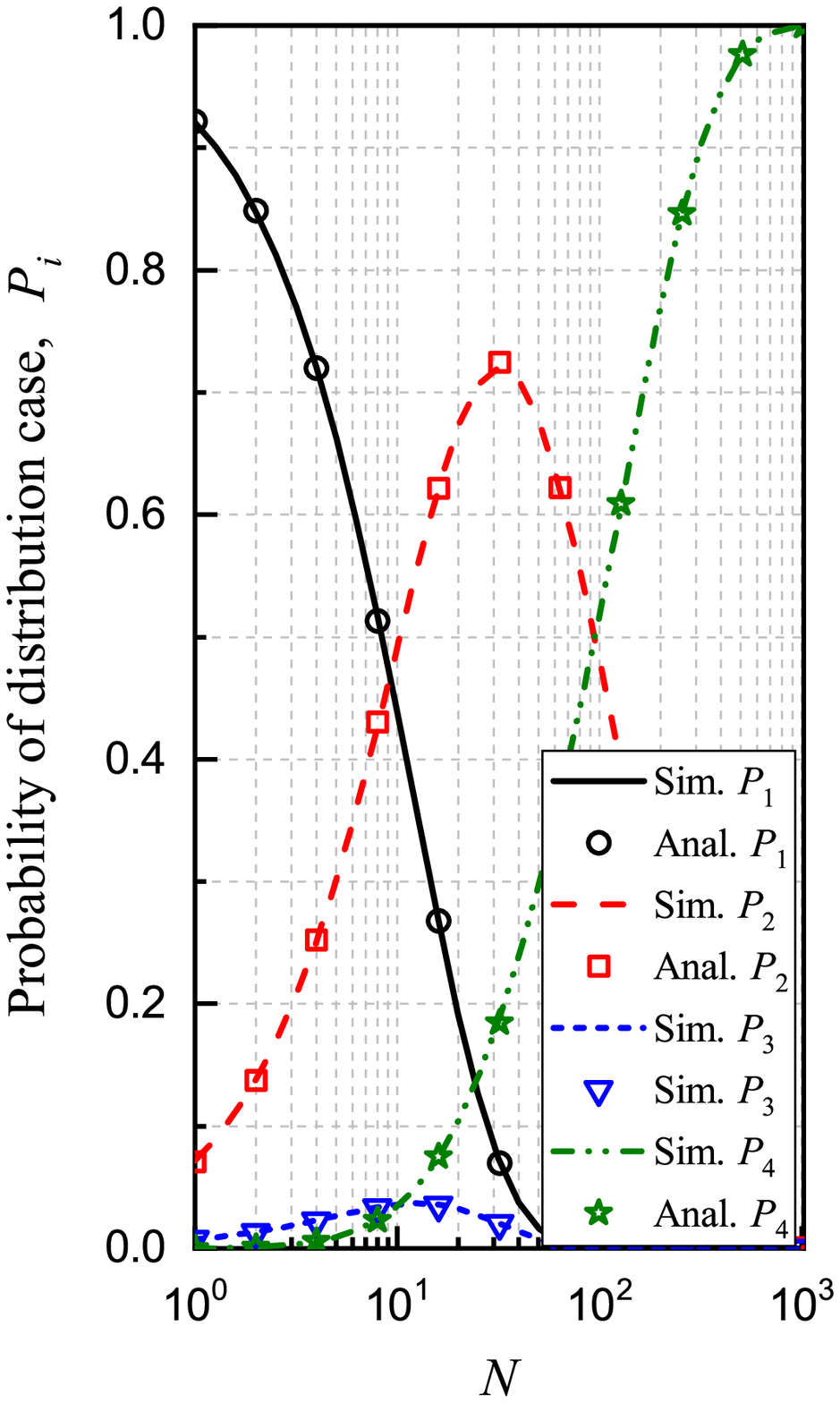}
}
\subfigure[Steerable-beam case.]{
\includegraphics[width=0.45\columnwidth]{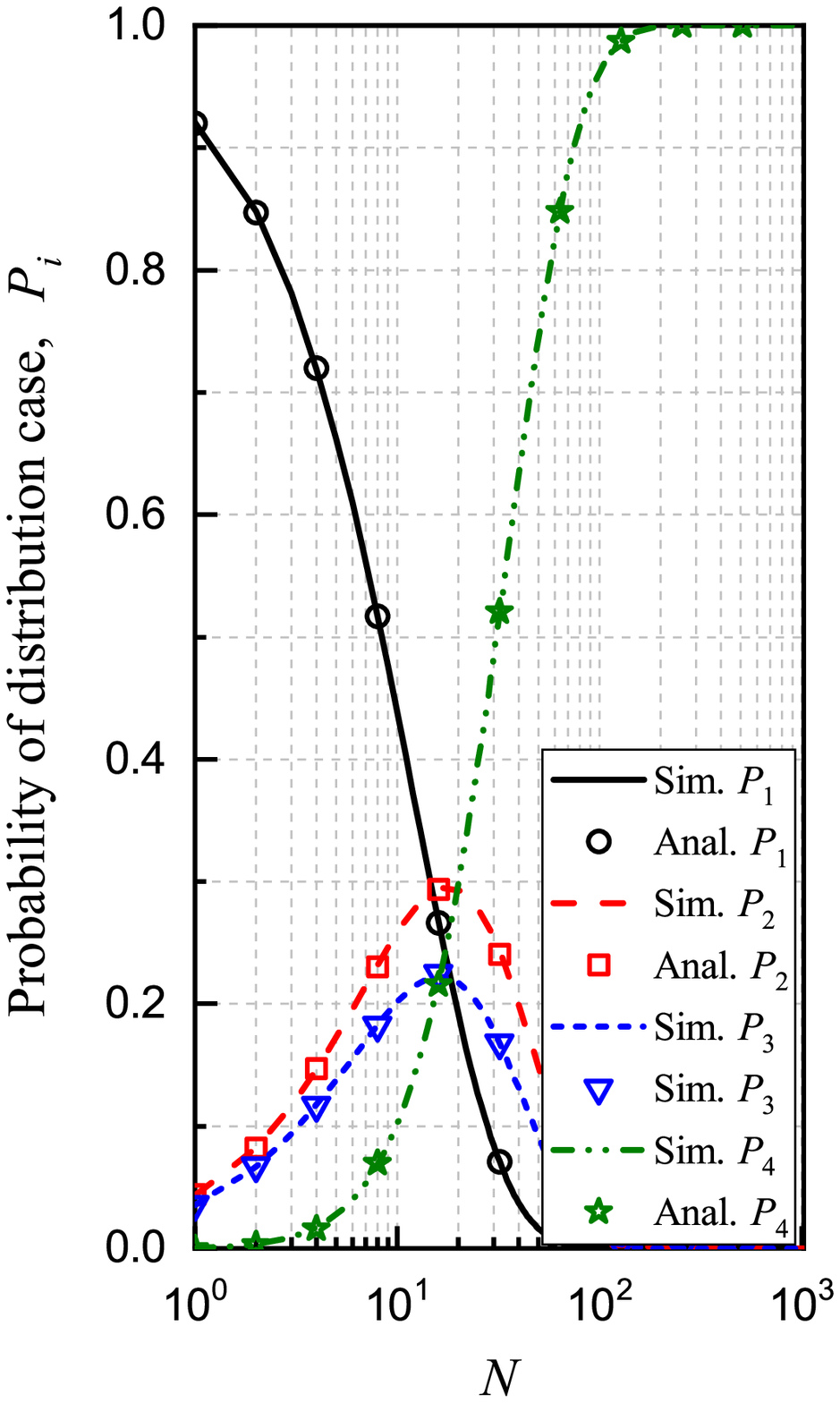}
}
\caption{Probabilities of distribution cases for the eavesdropping satellites with (a) fixed-beam and (b) steerable-beam antennas, where $\av=1200$ km, $\wth=40$ deg, and $\Delta\omega_\mathrm{sb}=15$ deg.}

\label{Fig:Pcase_vs_N}
\end{figure}

Fig. \ref{Fig:C_vs_as} shows the impact of the serving satellite's altitude $\as$ on the secrecy capacities with $\wth=20$ deg, $\wsb=10$ deg, $\epsilon=0.1$, $V=2$, $\{a_1, a_2\}=\{1015,1325\}$ km, and $\{N_1, N_2\}=\{78,220\}$ [\ref{Ref:Pachler}].
The analytical results for the satellites at different altitudes well match the simulation results.
The ergodic and outage secrecy capacities decrease with $\as$ because the path loss of the legitimate link becomes larger.
The secrecy capacities for the eavesdropping satellites with the FB antennas are always higher than those with the SB antennas.
This is mainly because when the eavesdropping satellites have a more flexible beam-steering capability, the number of eavesdropping satellites whose main lobes are directed to the terminal increases.

% as, thes, N, ae, wth, deltawsb, epsilon(or Rt)
\begin{figure}
\centering
\subfigure[Ergodic secrecy capacity.]{
\includegraphics[width=0.45\columnwidth]{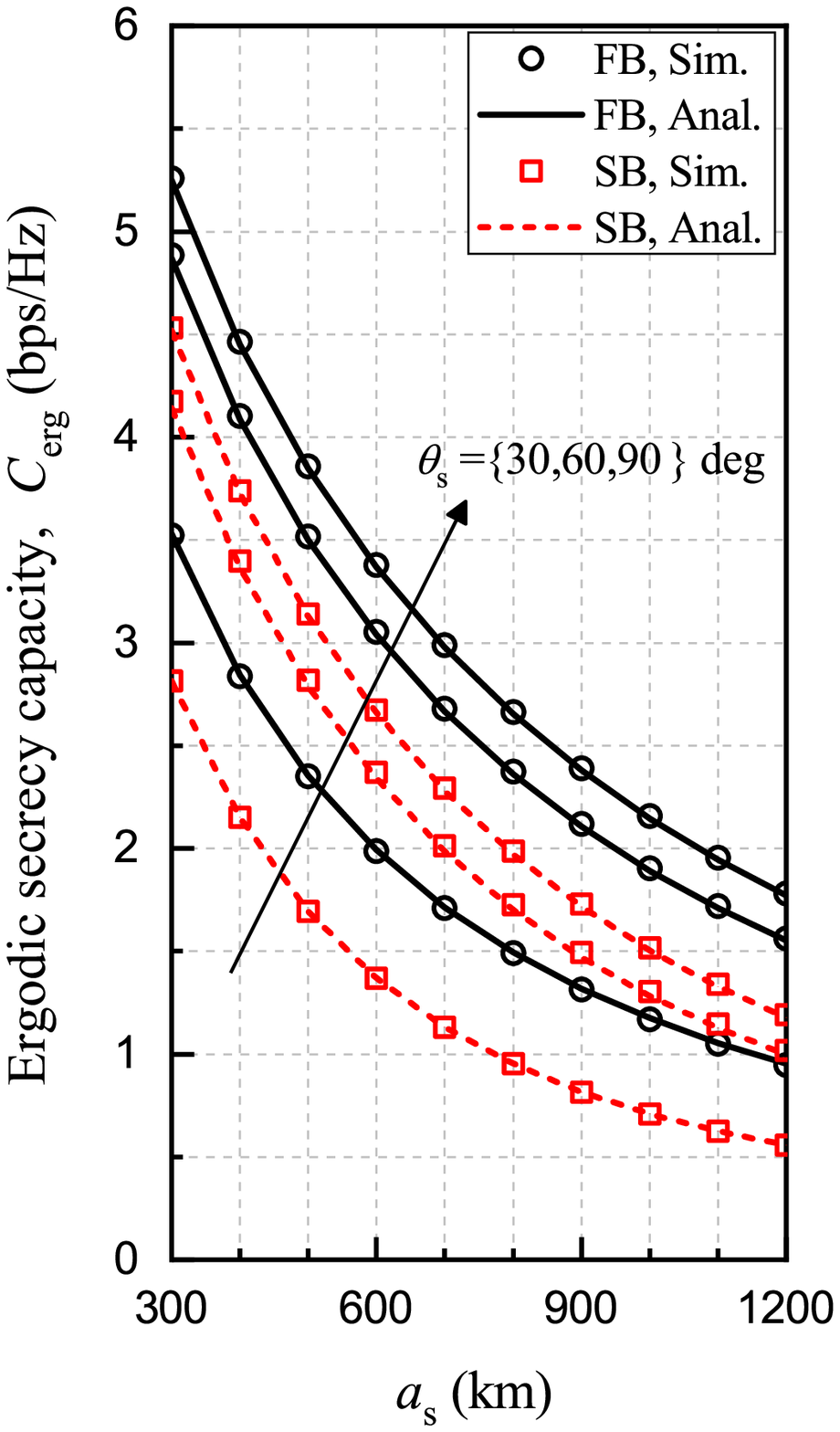}
}
\subfigure[Outage secrecy capacity.]{
\includegraphics[width=0.45\columnwidth]{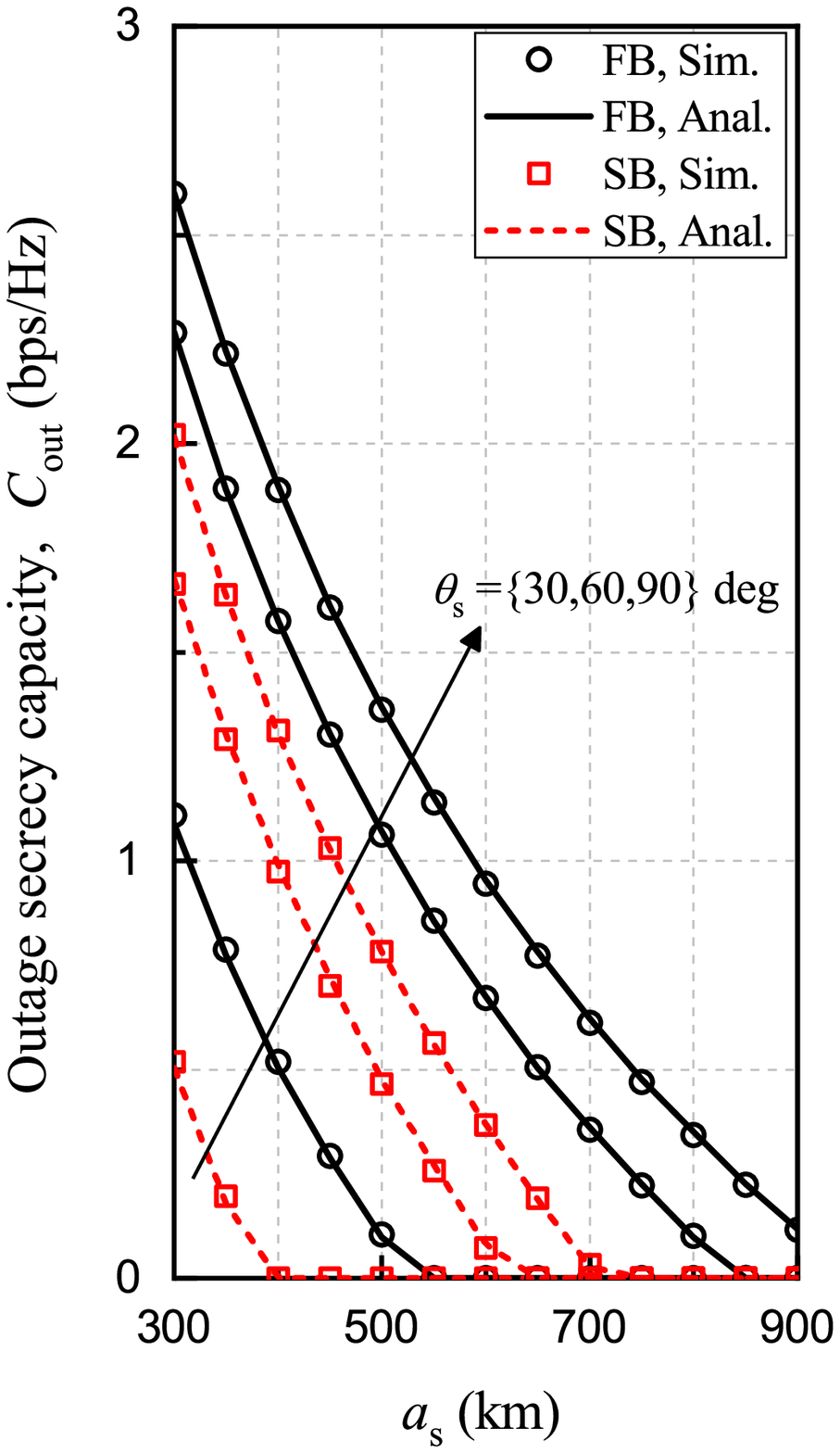}
}
\caption{Ergodic and outage secrecy capacities versus the altitude of the serving satellite $\as$ for various numbers of eavesdropping satellites with $V=2$, $\{a_1, a_2\}=\{1015,1325\}$ km, $\{N_1, N_2\}=\{78,220\}$, $\thes=60$ deg, $\wth=20$ deg, $\wsb=10$ deg, and $\epsilon=0.1$.}
\label{Fig:C_vs_as}
\end{figure}

% as, thes, N, ae, wth, deltawsb, epsilon(or Rt)
\begin{figure}
\begin{center}
\includegraphics[width=0.90\columnwidth]{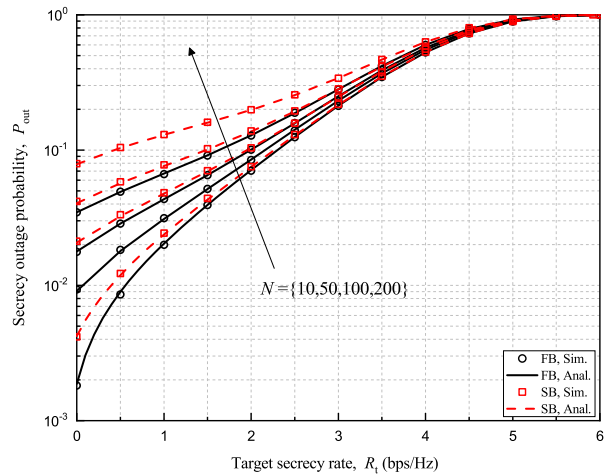}
\end{center}
\setlength\abovecaptionskip{.25ex plus .125ex minus .125ex}
\setlength\belowcaptionskip{.25ex plus .125ex minus .125ex}
\caption{Secrecy outage probability versus the target secrecy rate $\Rt$ for various numbers of eavesdropping satellites $N=\{10,50,100,200\}$ with  $\as=\av=600$, $\thes=60$ deg, $\wth=20$ deg, and $\wsb=10$ deg.}

\label{Fig:Pout_vs_Rt}
\end{figure}

Fig. \ref{Fig:Pout_vs_Rt} shows the secrecy outage probability versus the target secrecy rate $\Rt$ for various numbers of eavesdropping satellites $N=\{10,50,100,200\}$ with  $\as=\av=600$, $\thes=60$ deg, $\wth=20$ deg, and $\wsb=10$ deg. 
The analytical results for the exact secrecy outage probability are in good agreement with the simulation results. 
The secrecy outage probability increases with the target secrecy rate $\Rt$, which is expected from the definition of the secrecy outage probability. 
To achieve 10$\%$ of the secrecy outage probability for $N=\{10, 50, 100, 200\}$, the target secrecy rates for the FB case have to be less than approximately $\{2.3, 2.16, 1.99, 1.64\}$ bps/Hz.
For the SB case, the maximum values of the target secrecy rates are $\{2.26, 1.95, 1.48, 0.43\}$ bps/Hz, which are smaller than the FB case due to the greater ability of eavesdropping.
From these observation, SB antennas will enhance the eavesdropping performance at the price of higher costs.

Fig. \ref{Fig:Cerg_vs_N_add_asym} shows the ergodic secrecy capacity versus the number of eavesdropping satellites $N$ for various altitudes of the serving satellite $\as=\{300,600,1200\}$ km with $\thes=60$ deg, $\av=600$ km, $\wth=40$ deg, $\Delta\omega_\mathrm{sb}=20$ deg, and $\epsilon=0.1$.
As expected, the ergodic and outage secrecy capacities decrease with $N$ because the SNR at the most detrimental eavesdropping satellite increases with $N$.
This means that, from an eavesdroppers' perspective, it will be more efficient to deploy eavesdropping satellites as many as possible.
With small $N$, the performance gap between the FB and SB cases is marginal as the probability that the main lobes of the eavesdropping satellites' beams are directed to the terminal is very low for both FB and SB cases.
In contrast, with large $N$, the performance gap is large because there are many effective eavesdropping satellites that can steer their main lobes in the direction to the terminal for the SB case.

Fig. \ref{Fig:Pout_vs_N_add_asym} shows the secrecy outage probability versus the number of eavesdropping satellites $N$ for various altitudes of the serving satellite $\as=\{300,600,1200\}$ km with $\thes=60$ deg, $\av=600$ km, $\wth=40$ deg, $\Delta\omega_\mathrm{sb}=20$ deg, and $\Rt=3$ bps/Hz.
The secrecy outage probability increases with the number of eavesdropping satellites $N$.
For $\as=300$ km, when there are less than 56 eavesdropping satellites with the FB antennas, less than 10$\%$ of the secrecy outage probability is achieved, and, for the satellites with the SB antennas, it is achieved with less than  16 eavesdropping satellites.
On the other hand, when $\as=\{600,1200\}$ km, 10$\%$ of the secrecy outage probability cannot be achieved with any number of eavesdropping satellites due to the low SNR of the legitimate link.

% as, thes, N, ae, deltawsb, epsilon(or Rt)
\begin{figure}
\begin{center}
\includegraphics[width=0.90\columnwidth]{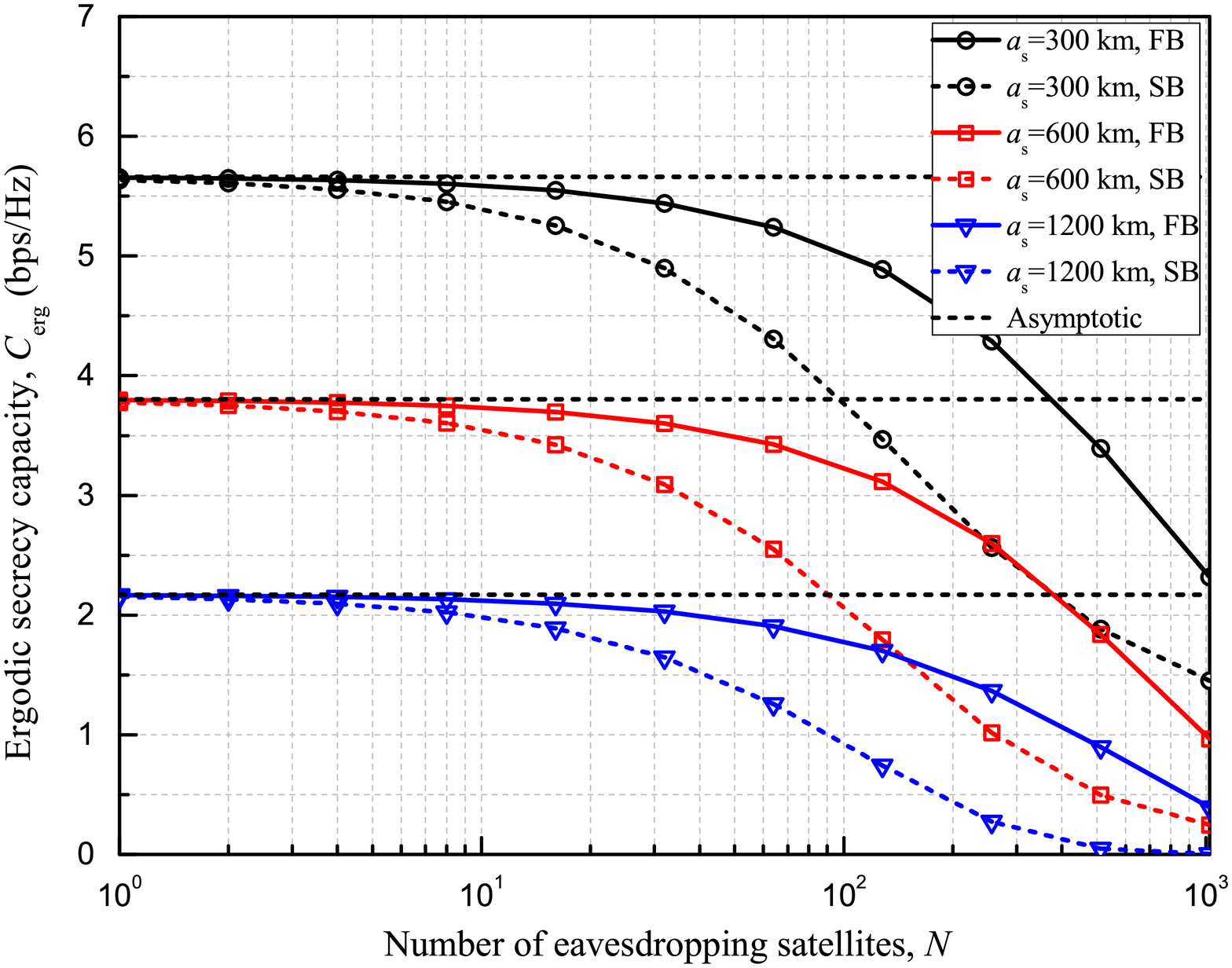}
\end{center}
\setlength\abovecaptionskip{.25ex plus .125ex minus .125ex}
\setlength\belowcaptionskip{.25ex plus .125ex minus .125ex}
\caption{Ergodic secrecy capacity  versus the number of eavesdropping satellites $N$ for various altitudes of the serving satellite $\as=\{300,600,1200\}$ km with $\thes=60$ deg, $\av=600$ km, $\wth=40$ deg, and $\Delta\omega_\mathrm{sb}=20$ deg.}
\label{Fig:Cerg_vs_N_add_asym}
\end{figure}

% as, thes, N, ae, deltawsb, epsilon(or Rt)
\begin{figure}
\begin{center}
\includegraphics[width=0.90\columnwidth]{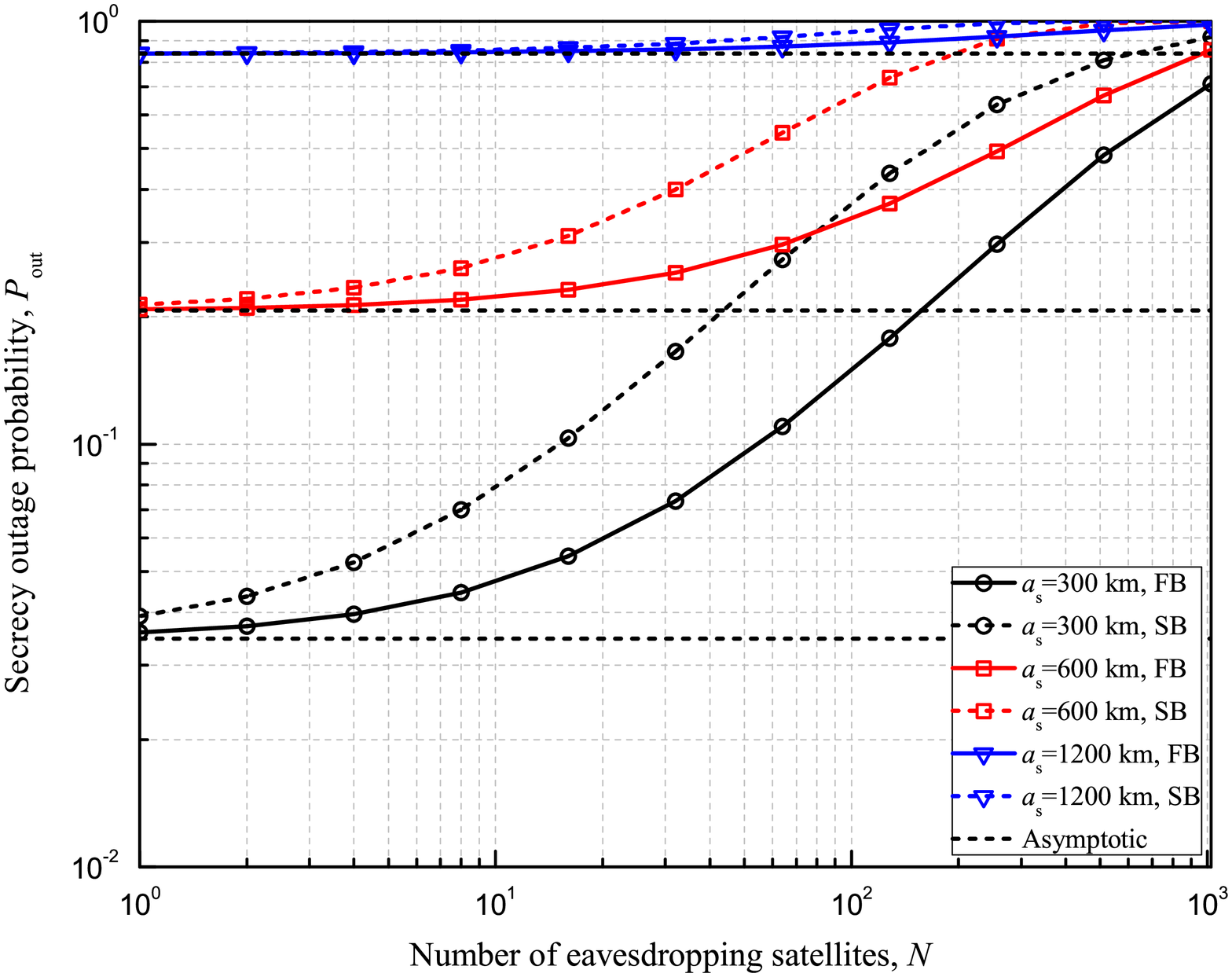}
\end{center}
\setlength\abovecaptionskip{.25ex plus .125ex minus .125ex}
\setlength\belowcaptionskip{.25ex plus .125ex minus .125ex}
\caption{Secrecy outage probability  versus the number of eavesdropping satellites $N$ for various altitudes of the serving satellite $\as=\{300,600,1200\}$ km with $\thes=60$ deg, $\av=600$ km, $\wth=40$ deg, $\Delta\omega_\mathrm{sb}=20$ deg, and $\Rt=3$ bps/Hz.}
\label{Fig:Pout_vs_N_add_asym}
\end{figure}

Fig. \ref{Fig:Cerg_vs_wsb} shows the ergodic secrecy capacity  versus the steerable angle of the boresight $\wsb$ for various altitudes of the eavesdropping satellites $\av=\{600,800,1200\}$ km with $\as=600$ km, $\thes=60$ deg, $\wth=40$ deg, and $N=100$. 
As $\wsb$ increases, the ergodic secrecy capacity for the SB case decreases because the surface area of $\Aeml$ increases.
When $\wsb$ is larger than $\wsb^*$ given in Remark \ref{Rem:wsb*}, the surface area of $\Aeml$ becomes the same as that of $\Ae$, which is the reason why the ergodic secrecy capacity stays constant.
With small $\wsb$, as $\av$ increases, the ergodic secrecy capacities for both FB and SB cases decrease because the number of effective eavesdropping satellites increases with $\av$.
More specifically, the average numbers of effective eavesdropping satellites are given by $\mathbb{E}[\BPP(\Ae)]=N\SAe/\SA$, which results in $\{\mathbb{E}[\BPP(\Ae)]\}=\{4.3,5.57,7.92\}$ for $\av=\{600,800,1200\}$ km and $N=100$.
In contrast, with large $\wsb$, the secrecy capacity for the SB case increases with $\av$.
This is because when the eavesdropping satellites have the greater beam-steering ability, the number of satellites whose main lobe is directed to the terminal increases, which makes up the decreased number of effective eavesdropping satellites for small $\av$. Thus, the impact of the increased SNRs of the eavesdropping links from smaller path loss with small $\av$ becomes a dominant factor to degrade the ergodic secrecy capacity.

% as, thes, N, ae, deltawsb, epsilon(or Rt)
\begin{figure}
\begin{center}
\includegraphics[width=0.90\columnwidth]{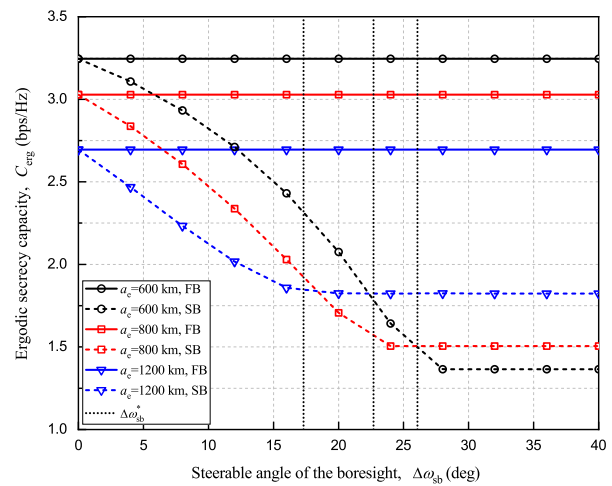}
\end{center}
\setlength\abovecaptionskip{.25ex plus .125ex minus .125ex}
\setlength\belowcaptionskip{.25ex plus .125ex minus .125ex}
\caption{Ergodic secrecy capacity  versus the steerable angle of the boresight $\wsb$ for various altitudes of the eavesdropping satellites $\av=\{600,800,1200\}$ km with $\as=600$ km, $\thes=60$ deg, $\wth=40$ deg, and $N=100$.}
\label{Fig:Cerg_vs_wsb}
\end{figure}

\section{Conclusions}\label{Sec:Conclusions}
In this paper, we investigated uplink low Earth orbit satellite communication systems with eavesdropping satellites randomly distributed according to binomial point processes (BPPs).
The possible distribution cases and the distance distributions for the eavesdropping satellites were analyzed based on the characteristics of the BPP.
The distributions of the signal-to-noise ratios (SNRs) at both the serving and the most detrimental eavesdropping satellite were derived in closed-forms.
The ergodic and outage secrecy capacities with the secrecy outage probability were analyzed by using the SNR and distance distributions.
To reduce the computational complexity of the performance evaluation, the approximate expressions were derived with the help of the Poisson limit theorem, and the asymptotic performance was obtained in various scenarios.
The analyses were verified by the simulation results where
the impact of the location of the serving satellite, the number of eavesdropping satellites, and the beam-steering capability was also discussed in terms of secure communications.
The analytical results are expected to give a guideline when designing practical techniques for secure satellite communication systems, e.g., multi-antenna transmissions and anti-jamming schemes.

\appendices

\section{Proof of Corollary \ref{Cor:prob_pq}}\label{App:Cor1}
Let $g$ be the cap height of a spherical cap $\mathcal{X}$ with radius $r$. Then, the surface area of $\mathcal{X}$ is given by $\mathcal{S}_{\mathcal{X}}=2\pi r g$ [\ref{Ref:Book:Polyanin}].
The region $\Ae$ is a spherical cap with the radius $r+\av$ and cap height $\av$, of which surface area is given by 
\begin{align}\label{eq:area_Ae}
    \SAe=2\pi(r+\av)\av.
\end{align}
As shown in Fig. \ref{Fig:area_description}, the region $\Aeml$ is a spherical cap with the radius $r+\av$ whose cap height $\overline{\mathrm{NM}}$ is  $\overline{\mathrm{ON}}-\overline{\mathrm{OM}}= (r+\av)(1 - \cos\psth)$.
The surface area of $\Aeml$ is given by 
\begin{align}\label{eq:area_Aeml}
\SAeml
    = 2\pi(r+\av)\overline{\mathrm{NM}}=2\pi(r+\av)^2 (1 - \cos\psth).    
\end{align}
The surface area of the region $\Aesl$ can be obtained from the difference between $\SAe$ and $\SAeml$ as
\begin{align}\label{eq:area_Aesl}
\SAesl
    &= \SAe-\SAeml
    = 2\pi(r+\av)\{(r+\av)\cos\psth-r\}. 
\end{align}

For the three disjoint regions $\Aeml$, $\Aesl$, and $\Ae^{\mathrm{c}}$ with $\Aeml \cup \Aesl \cup \Ae^{\mathrm{c}}=\A$, the probability that $p$ and $q$ points respectively lie in $\Aeml$ and $\Aesl$ is obtained by the finite-dimensional distribution as
\begin{align}\label{eq:prob_pq}
    \mathbb{P}[&\BPP(\Aeml)=p,\BPP(\Aesl)=q,\BPP(\Ae^{\mathrm{c}})=N\!-\!p\!-\!q]\nonumber\\
    &= \frac{N!}{p!q!(N\!-\!p\!-\!q)!}\!\left(\frac{\mathcal{S}_{\Aeml}}{\mathcal{S}_{\A}}\right)^p\! \left(\frac{\mathcal{S}_{\Aesl}}{\mathcal{S}_{\A}}\right)^q\! \left(\frac{\mathcal{S}_{\Ae^{\mathrm{c}}}}{\mathcal{S}_{\A}}\right)^{N-p-q}.
\end{align}
By plugging \eqref{eq:area_Ae}-\eqref{eq:area_Aesl} into \eqref{eq:prob_pq} with the fact that $\S_{\Ae^{\mathrm{c}}}=\S_{\A}-\S_{\Ae}$, the proof is complete.
%  of Corollary 1

\section{Proof of Lemma 2}\label{App:Lem2}
Let $\mathcal{A}(x)$, $x \in [\av, 2r+\av]$, be the surface of a spherical cap with the radius $r+\av$ and the cap height $g(x)$ such that the distance between any point on $\mathcal{A}(x)$ and the terminal is less than $x$.
Then, the surface area of $\mathcal{A}(x)$ is given by
\begin{align}\label{A_c}
\mathcal{S}_{\mathcal{A}(x)}
    = 2\pi(r+\av)g(x)=\frac{\pi(r+\av)(x^2-\av^2)}{r},
\end{align}
where $g(x)=(x^2-\av^2)/(2r)$ is obtained using the Pythagorean theorem.
For example, when $x=\av$, the surface area vanishes, i.e., $\mathcal{S}_{\mathcal{A}(\av)}=0$, and when $x=2r+\av$, the area $\mathcal{A}(x)$ becomes the entire sphere with the radius $r+\av$ whose surface is given by $\mathcal{S}_{\mathcal{A}(2r+\av)}=4\pi(r+\av)^2$.
The number of points on $\mathcal{A}(x)$ follows the binomial distribution with the total number of points $N$ and the success probability $p(\mathcal{A}(x))$, given by
\begin{align}\label{eqn:p(x)}
p(\mathcal{A}(x))
    = \frac{\mathcal{S}_{\mathcal{A}(x)}}{\mathcal{S}_{\mathcal{A}(2 r+\av)}}=\frac{x^2-\av^2}{4r (r+\av)}.
\end{align}
Since the success probability is the probability that an eavesdropping satellite is located in $\mathcal{A}(x)$, it is equivalent to the probability that the distance between the terminal and eavesdropping satellite $\idxe$ is less than $x$.
With this fact, the proof is complete.

\section{Proof of Lemma 3}\label{App:Lem3}
By the law of cosines, the distance between the terminal and the eavesdropping satellite $\idxe$ is obtained as $d_\idxe=\sqrt{r^2+(r+\av)^2-2r(r+\av)\cos\psi}$. Since $d_\idxe$ is an increasing function of $\psi \in [0,\psmax]$, the probability that the eavesdropping satellite $\idxe$ lies in $\Aeml$ can be expressed as $\Pr[\idxe \in \BPPeml]=\Pr[d_\idxe \le \dth]$, where $\dth=\sqrt{r^2+(r+\av)^2-2r(r+\av)\cos\psth}$. 
From this fact, the CDF of $X_\idxe$ is given by
\begin{align}\label{eq:CDFXe1}
F_{X_\idxe}(x)
    =& \Pr [X_\idxe \le x] \nonumber\\
    =&  \Pr[d_\idxe \le x|\idxe \in \BPPeml]\nonumber\\
    =&  \frac{\Pr[d_\idxe \le x, d_\idxe \le \dth]}{\Pr[d_\idxe \le \dth]}\nonumber\\
    =& \begin{cases} 0, & \mbox{if } x \le \av,\\ 
     \frac{F_{d_\idxe}(x)}{F_{d_\idxe}(\dth)}, & \mbox{if } \av < x \le \dth, \\
     1, & \mbox{if } x > \dth.
     \end{cases}
\end{align}
Plugging the result of Lemma \ref{Lem:CDF_d_e} into \eqref{eq:CDFXe1} completes the proof.

\section{Proof of Lemma 8}\label{App:Lem8}
Given $\BPP(\Aeml)=p$, the CDF of the SNR at the most detrimental eavesdropping satellite in $\Aeml$ is given by
\begin{align}\label{eq:CDF_SNReml1}
{F_{\gemdml}}(x) 
    &= \mathbb{P}[\gemdml \le x]\nonumber\\
    &= \mathbb{E}_{X_\idxe}\left[\prod\limits_{\idxe\in\BPPeml}  {\mathbb{P} [ h_\idxe \le w_1 X_\idxe^{\alpha} x ]}  \right]\nonumber\\
    &\mathop=\limits^{(a)}  \left[ \int_{\av}^{\dth} F_{h_\idxe}(w_1 z^{\alpha} x) f_{X_\idxe}(z)dz \right]^p\nonumber\\
    &\mathop =\limits^{(b)} \left[\frac{2K}{\dth^2-\av^2}\sum\limits_{n = 0}^\infty  {\frac{{{{(m)}_n}{\delta ^n}{{(2b)}^{1 + n}}}}{{{{(n!)}^2}}}} \right.\nonumber\\
    &\qquad \times \left. \int_{\av}^{\dth}\gamma\left(1+n, \frac{w_1 z^{\alpha} x}{2b} \right) z dz\right]^{p},
\end{align}
where ($a$) follows form the fact that the distances of $p$ eavesdropping satellites, $X_\idxe$, $\idxe \in \BPPeml$, are independently and identically distributed, and ($b$) follows from the CDF of the channel  gain and the PDF of $X_\idxe$, given by \eqref{eq:CDF_ch_gain} and \eqref{eq:PDF_Xe}, respectively.

From the definition of the lower incomplete Gamma function, the integral in \eqref{eq:CDF_SNReml1} is expressed as
\begin{align}\label{eq:CDF_SNReml2_int1}
\mathcal{I}_{\mathrm{ml}}(x,n)
    &=\int_{\av}^{\dth}\int_0^{\Lambda(z)}t^n e^{-t} z dtdz,
\end{align}
where $\Lambda(z)=\frac{w_1 z^{\alpha} x}{2b}$.
Since $\Lambda(z)$ is an increasing function of $z$ for $z>0$ and $\alpha > 0$, the domain of the integration \eqref{eq:CDF_SNReml2_int1} can be divided into two domains $S_1$ and $S_2$ that are respectively given by
$S_1=\{0 \le t \le \Lambda(\av),\,\, \av \le z \le \dth\}$
and 
$S_2 =\{\Lambda(\av)\le t \le \Lambda(\dth),\,\, \Lambda^{-1}(t) \le z \le \dth\}$, where $\Lambda^{-1}(t)=\left(\frac{2bt}{w_1 x}\right)^{1/\alpha}$. 
The integral over domain $S_1$ is obtained as the product of two independent integrals over $z$ and $t$, which is given by
\begin{align}\label{eq:intS1_fin}
\mathcal{I}_{S_1}(x,n)
    =&\int_0^{\Lambda(\av)} \int_{\av}^{\dth}    {{t^n}e^{-t}} z dz dt  \nonumber\\
    =&\int_0^{\Lambda(\av)} {t^n}e^{-t} dt \times \int_{\av}^{\dth} z dz  \nonumber\\
    =& \frac{\dth^2-\av^2}{2} \gamma(1+n,\Lambda(\av)),
\end{align}
and the integral over domain $S_2$ is given by
\begin{align}\label{eq:intS2_fin}
\mathcal{I}_{S_2}(x,n)
    &=\int_{\Lambda(\av)}^{\Lambda(\dth)} {\int_{\Lambda^{-1}(t)}^{\dth} {{t^n}{e^{ - t}}z dz dt} } \nonumber\\
    &=\frac{1}{2}\int_{\Lambda(\av)}^{\Lambda(\dth)} {t^n}e^{ - t} \left(\dth^2-\left(\frac{2bt}{w_1 x}\right)^{\frac{2}{\alpha}}\right) dt  \nonumber\\
    &=\frac{\dth^2}{2}\!\!\int_{\Lambda(\av)}^{\Lambda(\dth)} {t^n}e^{ - t}dt -\frac{1}{2}\left(\frac{2b}{w_1 x}\right)^{\frac{2}{\alpha}} \!\!\int_{\Lambda(\av)}^{\Lambda(\dth)} \! t^{n+\frac{2}{\alpha}}e^{-t}  dt  \nonumber\\
    &=\frac{\dth^2}{2} \{\gamma(1+n,\Lambda(\dth))-\gamma(1+n,\Lambda(\av))\}\nonumber\\
    &\quad-\frac{1}{2}\left(\frac{2b}{w_1 x}\right)^{\frac{2}{\alpha}} \left\{\gamma\left(1+n+\frac{2}{\alpha},\Lambda(\dth)\right)\right.\nonumber\\
    &\quad\left.-\gamma\left(1+n+\frac{2}{\alpha},\Lambda(\av)\right)\right\}.
\end{align}
From \eqref{eq:CDF_SNReml1}-\eqref{eq:intS2_fin}, we can obtain the final expression for the CDF of $\gemdml$.

\section{Proof of Theorem 1}\label{App:Thm1}
Based on the law of total expectation, i.e., $\mathbb{E}[X]=\sum_i\mathbb{P}[A_i]\mathbb{E}[X|A_i]$, the ergodic secrecy capacity of the system is given by
\begin{align}\label{eq:Cerg1}
\Cerg
    % &= \mathbb{E}[R]\nonumber\\
    &=\sum_{p=0}^{N} \sum_{q=0}^{N-p} \mathcal{P}[N,p,q] \mathbb{E}[R\,|\,\BPP(\Aeml)=p,\BPP(\Aesl)=q] \nonumber\\
    &=\sum_{p=0}^{N} \sum_{q=0}^{N-p} \mathcal{P}[N,p,q]\nonumber\\
    &\quad\times\int_0^{\infty}\int_0^{x}\log_2\left(\frac{1+x}{1+y}\right) f_{\gs, {\g_{\mathrm{e}^*}^{(p,q)}}}(x,y) dy dx\nonumber\\
    &\mathop=\limits^{(a)}\sum_{p=0}^{N} \sum_{q=0}^{N-p} \frac{\mathcal{P}[N,p,q]}{\ln2} \int_0^{\infty}\frac{F_{\g_{\mathrm{e}^*}^{(p,q)}}(x)}{1+x}(1-F_\gs(x))dx,
\end{align}
where $f_{\gs, {\g_{\mathrm{e}^*}^{(p,q)}}}(x,y)$ is the joint PDF of $\gs$ and ${\g_{\mathrm{e}^*}^{(p,q)}}$, and ($a$) follows from the independence of two random variables $\gs$ and ${\g_{\mathrm{e}^*}^{(p,q)}}$.
From \eqref{eq:CDF_SNRs_fin}, \eqref{eq:CDF_SNRe*_fin}, and \eqref{eq:Cerg1}, the final expression for the ergodic secrecy capacity of the system can be obtained.

\section{Proof of Theorem 2}\label{App:Thm2}
Using the law of total probability, the secrecy outage probability is given by
\begin{align}\label{eq:Pout1}
\Pout
    &=\sum_{p=0}^{N} \sum_{q=0}^{N-p} \mathcal{P}[N,p,q] \mathbb{P}[\g_{\mathrm{e}^*}^{(p,q)} \geq  2^{-\Rt}(1+\gs)-1]\nonumber\\
    &=1-\sum_{p=0}^{N} \sum_{q=0}^{N-p} \mathcal{P}[N,p,q] \nonumber\\
    &\quad\qquad\times\int_{2^{\Rt}-1}^{\infty} F_{\g_{\mathrm{e}^*}^{(p,q)}}(2^{-\Rt}(1+x)-1) f_{\gs}(x)dx.
\end{align}
With the fact that $\frac{d\gamma(a,x)}{dx}=e^{-x}x^{a-1}$, the PDF of the SNR at the serving satellite in \eqref{eq:Pout1} is obtained by differentiating \eqref{eq:CDF_SNRs_fin} as \eqref{eq:PDF_SNRs}, which completes the proof.

\section{Proof of Lemma \ref{Lem:CDFgemlapp}}\label{App:Cor3}
When $\av \rightarrow 0$, the CDF of the SNR at the most detrimental eavesdropping satellite in $\Aeml$ is given by
\begin{align}\label{eq:CDF_SNRemlapp1}
{F_{\gemdmlapp}}(x) 
    \!=\! \mathbb{P}[\gemdmlapp \le x]
    \!= {\mathbb{E}_{{\PPPeml}}}\!\left[ {\prod\limits_{\idxet \in {\PPPeml}}\!\! {{F_{h_\idxet}}\!\left( {\frac{{{N_0}Wx}}{{P\ell(d_\idxet)}}} \right)} } \right].
\end{align}
By using the probability generating functional of the PPP $\PPPeml$ [\ref{Ref:Book:Chiu}]
\begin{align}\label{eq:PGFL}
    \mathbb{E}_{\PPPeml}\left[\prod\limits_{\idxet \in {\PPPeml}} f(x) \right]=\exp\left(-\lame \int_{\Aeml} (1-f(x))dx\right),
\end{align}
with conversion from the Cartesian to spherical coordinates, 
we have
\begin{align}\label{eq:CDF_SNRemlapp2}
{F_{\gemdmlapp}}&(x)
    = \exp \left(  - \lambda_{\mathrm{e}}{(r + \av)}^2\int_0^{2\pi } \int_0^{\psth} \right.\nonumber\\
    &\,\,\,\,\qquad\qquad\times \left. \left\{ {1 - {F_{h_\idxet}}\left( {\frac{{{N_0}Wx}}{{P\ell(d_\idxet)}}} \right)} \right\}\sin \psi d\psi d\varphi \right)\nonumber\\
    =& \exp \left( -\frac{N}{2} \left\{ 1-\cos\psth - K\sum\limits_{n = 0}^\infty  \frac{{{{(m)}_n}{\delta ^n}{{(2b)}^{1 + n}}}}{{{{(n!)}^2}}} \right.\right.\nonumber\\
    &\qquad\qquad\times\left.\left.\int_0^{{\psth}}  \int_0^{\Pi_{\mathrm{ml}}(x,\psi)} t^n e^{-t} \sin\psi dt d\psi\right\}  \right),
\end{align}
where $\Pi_{\mathrm{ml}}(x,\psi)$ is in \eqref{eq:Pi_ml}.
The integration in \eqref{eq:CDF_SNRemlapp2} can be obtained by using the similar way as in the proof of Lemma \ref{Lem:CDF_SNReml}. The details are omitted due to the space limitation.

\section{Proof of Corollary \ref{cor:Cerg_N0}}\label{App:CorCerg_N0}
As $N$ goes to zero, $F_{\gemdapp}(x)$ becomes one. Substituting $F_{\gemdapp}(x)=1$ into \eqref{eq:Cerg_assym_fin}, we have
\begin{align}\label{eq:Cerg_N01}
\Cerg
    &\rightarrow \frac{1}{\ln2} \int_0^{\infty}\frac{1-F_\gs(x)}{1+x}dx\nonumber\\
    &\mathop\approx\limits^{(a)}\frac{K}{\ln2}\sum_{k=0}^{\floor{m}-1}\sum_{t=0}^{k}\frac{(-1)^k (1-m)_{k} \delta^k (w_1 \ds^{\alpha})^t}{k! t! \left(\frac{1}{2b}-\delta\right)^{k-t+1}} \nonumber\\
    &\quad\qt\int_0^{\infty}\frac{x^t}{1+x}e^{-\left(\frac{1}{2b}-\delta\right)w_1 \ds^{\alpha}x}dx,
\end{align}
where ($a$) follows from the fact that the CDF of the channel gain in \eqref{eq:CDF_ch_gain} can be approximated with integer $m$ as [\ref{Ref:Miridakis}]
\begin{align}
F_{h}(x)=1-K\sum_{k=0}^{m-1}\sum_{t=0}^{k}\frac{(-1)^k (1-m)_{k} \delta^k }{k! t! \left(\frac{1}{2b}-\delta\right)^{k-t+1}} x^t e^{-\left(\frac{1}{2b}-\delta\right)x}.
\end{align}
Using $\int_{0}^{\infty}\frac{x^t}{1+x}e^{-\xi x}=e^\xi\Gamma(1+t)\Gamma(-t,\xi)$ in \eqref{eq:Cerg_N01}, we obtain the final expression of the ergodic secrecy capacity.
Similarly, the asymptotic secrecy outage probability in \eqref{eq:Pout_N0} is obtained by simply letting $F_{\gemdapp}(x)=1$ in \eqref{eq:Pout_assym_fin}.

\section{Proof of Theorem \ref{thm:CergHighSNR}}\label{App:CorCergHighSNR}
Let $\mathrm{e}_0 \in \BPPe$ denote the nearest eavesdropping satellite from the terminal. Then, the SNR at the nearest eavesdropping satellite is given by $\g_{\mathrm{e}_0} = \frac{P h_{\mathrm{e}_0} \ell(d_0)}{N_0 W}$, where $d_0=\min_{\idxe \in \BPPe}d_\idxe$ is the distance between the nearest eavesdropper and the terminal. The CDF of $d_0$ is given by 
\begin{align}\label{eq:cdfd0}
F_{d_0}(x)&=\P\left[\min\limits_{\idxe\in\BPPe} d_\idxe \le x\right]=1-\prod\limits_{\idxe\in\BPPe}\P[d_\idxe > x]\nonumber\\
&=1-(1-F_{d_\idxe}(x))^N.
\end{align}
By differentiating \eqref{eq:cdfd0} with respect to $x$ and using \eqref{eq:CDFde}, the PDF of $d_0$ is given by
\begin{align}
f_{d_0}(x)
    &=N(1-F_{d_\idxe}(x))^{N-1}f_{d_\idxe}(x)\nonumber\\
    &=\frac{N x}{2r(r+\av)}\left(1-\frac{x^2-\av^2}{4r(r+\av)}\right)^{N-1}.
\end{align}

In the high-SNR regime, the ergodic secrecy capacity is upper-bounded by 
\begin{align}
\Cerg
    &\le \E\left[\log_2\left(\frac{1+\gs}{1+\g_{\mathrm{e}_0}}\right)\right]\mathop\approx\limits^{P\rightarrow\infty} \E\left[\log_2\left(\frac{\gs}{\g_{\mathrm{e}_0}}\right)\right] \delequal \Cerg^{\infty}.
\end{align}
Using the law of total probability, the upper-bound is expressed as \begin{align}\label{eq:CergUB2}
\Cerg^{\infty}
    =&\P[d_0 \le \dth]\underbrace{\E\left[\log_2\left.\left(\frac{\gs}{\g_{\mathrm{e}_0}}\right)\right|d_0 \le \dth\right]}_{\Xi_1} \nonumber\\
    &+ \P[\dth < d_0 \le \dmax]\underbrace{\E\left[\log_2\left.\left(\frac{\gs}{\g_{\mathrm{e}_0}}\right)\right| \dth < d_0 \le \dmax\right]}_{\Xi_2} \nonumber\\
    &+ \P[d_0 > \dmax]\underbrace{\E[\log_2\gs|d_0 > \dmax]}_{\Xi_3}.
\end{align}
The first expectation on the right-hand side of \eqref{eq:CergUB2}, $\Xi_1$, can be further approximated as
\begin{align}\label{eq:Xi1}
\Xi_1
    &= \E_{h_{\mathrm{s}}, h_{\mathrm{e}_0}, d_0}\left[\log_2\left.\left(\frac{\gs}{\g_{\mathrm{e}_0}}\right)\right|d_0 \le \dth\right]\nonumber\\
    &\mathop\approx\limits^{(a)} \E_{d_0}\left[\log_2\left.\left(\frac{\E_{h_{\mathrm{s}}}[\gs]}{\E_{h_{\mathrm{e}_0}}[\g_{\mathrm{e}_0}]}\right)\right|d_0 \le \dth\right] \nonumber\\
    &\mathop=\limits^{(b)} \E_{d_0}\left[\log_2\left.\left(\frac{\ds^{-\alpha}}{d_0^{-\alpha}}\right)\right|d_0 \le \dth\right]\nonumber\\
    &\mathop\approx\limits^{(c)} \log_2\left(\frac{\ds^{-\alpha}}{\E_{d_0}[d_0|d_0 \le \dth]^{-\alpha}}\right),
\end{align}
where ($a$) follows from the results in Appendix C of [\ref{Ref:Yuan2}] and the independence between the random variables $h_{\mathrm{s}}$ and $h_{\mathrm{e}_0}$; ($b$) follows from the definition of the SNR in \eqref{eq:SNR}; and ($c$) follows from the approximation given in Appendix B of [22].

To further derive \eqref{eq:Xi1}, we obtain the conditional CDF $F_{d_0|d_0\le \dth}(x)$ as
\begin{align}
F_{d_0|d_0\le \dth}(x)
    &= \P[d_0 \le x | d_0\le \dth]\nonumber\\
    &= \begin{cases} 0, & \mbox{if } x \le \av,\\ 
     \frac{F_{d_0}(x)}{F_{d_0}(\dth)}, & \mbox{if } \av < x \le \dth, \\
     1, & \mbox{if } x > \dth,
     \end{cases}
\end{align}
which gives the corresponding PDF as
\begin{align}
f_{d_0|d_0\le \dth}(x)
    = \begin{cases} \frac{f_{d_0}(x)}{F_{d_0}(\dth)}, & \mbox{if } \av < x \le \dth, \\
     0, & \mbox{otherwise}.
     \end{cases}
\end{align}
Using this PDF, $\E_{d_0}[d_0|d_0 \le \dth]$ in \eqref{eq:Xi1} is given by
\begin{align}\label{eq:Xi1exp}
&\E_{d_0}[d_0|d_0 \le \dth]
    = \int_{\av}^{\dth} x f_{d_0|d_0\le \dth}(x) dx\nonumber\\
    &=\frac{N}{2r(r+\av)F_{d_0}(\dth)}\int_{\av}^{\dth} x^2 \left(1-\frac{x^2-\av^2}{4r(r+\av)}\right)^{N-1} dx\nonumber\\
    &\mathop=\limits^{(a)}\frac{N}{2r(r+\av)F_{d_0}(\dth)}\sum\limits_{i=0}^{N-1}\binom{N-1}{i}\nonumber\\
    &\times\left(1+\frac{\av^2}{4r(r+\av)}\right)^{N-1-i}\left(\frac{-1}{4r(r+\av)}\right)^i\frac{\dth^{2i+3}-\av^{2i+3}}{2i+3},
\end{align}
where ($a$) follows from the binomial expansion. From \eqref{eq:Xi1} and \eqref{eq:Xi1exp}, we can obtain $\Xi_1$. Similarly, $\Xi_2$ and $\Xi_3$ can be readily derived, but the details are omitted due to lack of space. Plugging $\Xi_1$, $\Xi_2$, and $\Xi_3$ into \eqref{eq:CergUB2} with the CDF of~$d_0$ completes the proof.

\ifCLASSOPTIONcaptionsoff
  \newpage
\fi

\end{document}